\documentclass[oribib]{llncs}
\usepackage{amsmath}

\newcommand{\ignore}[1]{}

\usepackage{times}
\usepackage{helvet}
\usepackage{courier}
\usepackage{color}
\usepackage{multicol}

\usepackage{amssymb}
\usepackage{graphicx}
\usepackage{color}
\usepackage{hhline}

\ignore{++++++++++++++
\documentclass[a4paper,UKenglish]{lipics}%{lipics-v2018}

\usepackage{microtype}%if unwanted, comment out or use option "draft"

\bibliographystyle{plainurl}

\usepackage{multicol}
\usepackage{amsmath}
\usepackage{amssymb}
\usepackage{graphicx}
\usepackage{smallsec}
\usepackage{float}
\usepackage{epsfig}
\usepackage{color}
\usepackage{hhline}
}

\usepackage{pdfsync}

\newcommand{\boxtheorem}{\hfill $\Box$}
\newcommand{\nit}[1]{{\it #1}}

\usepackage{times}

\usepackage{graphicx}

\usepackage{algpseudocode}
\usepackage{algorithm}
\usepackage{epstopdf}
\usepackage{enumerate}
\usepackage{times}

\newcounter{lemmaA-counter}
\newcounter{propositionA-counter}

{\vskip \abovedisplayskip \refstepcounter{lemmaA-counter}%
\noindent {\bf Lemma A.\arabic{lemmaA-counter}.}}%

{\vskip \abovedisplayskip \refstepcounter{propositionA-counter}%
\noindent {\bf Proposition A.\arabic{propositionA-counter}.}}%

\newcommand{\defproof}[2]{{\noindent\bf Proof of #1:\
}#2 \boxtheorem}

\newcommand{\mc}[1]{\mathcal{ #1}}

\newcommand{\bcq}{BCQ}
\newcommand{\cq}{CQ}

\newcommand{\mf}[1]{\mathfrak{ #1}}
\newcommand{\n}{~{\it not}~}

\newcommand{\nn}{\nit{null}}

\newcommand{\fo}{FO}

\newcommand{\red}[1]{\textcolor{red}{#1}}
\newcommand{\blue}[1]{\textcolor{blue}{#1}}

\newcommand{\sfd}{{\sf d}}
\newcommand{\sfs}{{\sf s}}

\newcommand{\srep}{\nit{Rep}^{\sf S}(D,\Sigma)}
\newcommand{\crep}{\nit{Crep}(D,\Sigma)}

\ignore{\volumeinfo{\ignore{editors here}}
{2}
{}
{31}
{1}
{1}
\EventShortName{XXX}
++++++++++++++++}

\title{\vspace*{-1cm} {\bf Repair-Based Degrees of Database Inconsistency: Computation and Complexity}}
\author{{\bf Leopoldo Bertossi}\thanks{Member of the ``Millenium Institute for Foundational Research on Data" (IMFD, Chile). \ Email: bertossi@scs.carleton.ca. \ignore{ Research supported by NSERC Discovery Grant \#06148.}}\\
RelationalAI Inc. \ and \ Carleton University,  Canada \vspace{-5mm}}

\ignore{\author{{\bf Leopoldo Bertossi\vspace{-3mm}}\thanks{Member of the ``Millenium Institute for Foundational Research on Data" (IMFD, Chile). \ Email: \texttt{bertossi@scs.carleton.ca}. \ Research supported by NSERC Discovery Grant \#06148.} \ }\thanks{ Carleton Univ., \
School of Computer Science, Canada. \
bertossi@scs.carleton.ca.}
\institute{RelationalAI Inc. \ and \ Carleton University,  Canada}}
\institute{}

\titlerunning{Repair-Based Degrees of Database Inconsistency}
\authorrunning{Bertossi}

\begin{document}
\pagestyle{plain}
\maketitle
\thispagestyle{empty}

\begin{abstract}We propose a generic numerical measure of the inconsistency of a database with respect to a set of integrity constraints. It is based on an abstract repair semantics. In particular, an inconsistency measure associated to cardinality-repairs is investigated in detail. More specifically, it is shown that it can be computed via answer-set programs, but sometimes its computation can be intractable in data complexity. However, polynomial-time  deterministic and randomized approximations are exhibited. The behavior of this measure under small updates is analyzed, obtaining fixed-parameter tractability results. Furthermore, alternative inconsistency measures are proposed and discussed.
\end{abstract}

%\keywords{Integrity constraints in databases, inconsistent databases, database repairs, inconsistency measures}

\ignore{\begin{abstract}In this work, ... \vspace{-2mm}
\end{abstract}
}

\section{Introduction}
\vspace{-2mm}

Intuitively, a relational database may be more or less consistent than other databases for the same schema and with the same integrity constraints (ICs). This comparison can be accomplished by assigning a {\em measure of inconsistency} to a database. The associated inconsistency degree of a database $D$ with respect to (wrt.) a set of ICs $\Sigma$ should depend on how complex it is to restore consistency; or more technically, on the class of {\em repairs} of $D$ wrt. $\Sigma$.  \ Accordingly, our take on this issue is that a degree of inconsistency depends upon a repair semantics, and then,
on how consistency is restored. This implies that a degree of inconsistency involves both the admissible repair actions and how close we want stay to the instance at hand. \ To achieve this, we can apply concepts and results about database repairs  (cf. \cite{Bertossi2011} for a survey and references).

The problem of measuring inconsistency has been investigated mostly by the knowledge representation community, but scarcely by the data management community. Furthermore, the approaches and results obtained in KR do not immediately apply or do not address the problems that are natural and relevant in databases, such as their computation and complexity in terms of the size of the database (i.e. data complexity). Actually, several (in)consistency measures have been considered in knowledge representation \cite{hunter,thimm,vanina}, mostly for propositional knowledge bases, or have been applied with grounded first-order representations, obtaining in essence a propositional representation. \ It becomes interesting to consider inconsistency measures that are closer to database applications, and whose formulation and computation stay at the relational level.

In this work we investigate possible ways to make these ideas concrete, by defining and analyzing  a generic class of repair-based measures of inconsistency of relational database instances. For a particular and natural inconsistency measure in this class we  provide a computational mechanism that uses  {\em answer-set programming} (ASP) \cite{brewka}, also known as {\em logic programming with stable model semantics} \cite{gl91}. We also provide some first results on the complexity of computing this measure. It turns out that ASPs provide the exact expressive and computational power needed to compute this measure.

The particular inconsistency measure we investigate in more depth here is motivated by one used before  to measure the degree of satisfaction of {\em functional dependencies} in a relational database \cite{mannila}. We extend and reformulate it in terms of database repairs, applying it to the larger class of {\em denial constraints} \cite{Bertossi2011}. Actually, it can be naturally applied to any class of monotonic ICs (in the sense that as the database grows only more violations can be added); and  also with other non-monotonic classes of ICs, such as inclusion- and tuple-generating  dependencies, as long as we repair only through tuple deletions. However, the measure can be redefined  using the symmetric difference between the original database and the repairs when tuple insertions are also allowed as repair actions.

The  investigation we carry out of the particular inconsistency measure is, independently from possible alternative measures, interesting {\em per se}: We stay at the relational (or first-order) level (as opposed to the propositional case usually considered in knowledge representation) and we stress computability and complexity issues in terms of the size of the database. This  provides a pattern for the investigation of other possible consistency measures, along similar lines. We are not aware of research that emphasizes computational aspects of inconsistency measures; and we start filling in this gap here.  It is likely  that other possible consistency measures in the relational setting are also polynomially-reducible to the one we investigate here (or the other way around), and results for one can be leveraged for the other(s). This is a matter of future research.

It  is natural to try to have a quantitative sense for the level of inconsistency that may be present in a large database. From this point of view, the inconsistency measure can be seen as a complex aggregation we may want to compute exactly or approximately. Our measure addresses such a need, and also opens the ground for counterfactual  analysis
of the data, in the direction of determining how the inconsistency degree changes under certain, possibly hypothetical, updates, much in the spirit of causality in databases \cite{suciu,tocs}.\footnote{The connection between database causality and database repairs was established and exploited for causality purposes in \cite{tocs,foiks18}.} Furthermore, this measure can be used as a basis for developing sampling techniques for estimating the inconsistency degree of a database. We give first steps in all these directions.

The kind of results that we obtain in terms of computation and complexity are extendible to other, broader logic-based settings, such as ontologies and knowledge bases, and  in particular, to {\em ontology-based data access}  (OBDA) \cite{xiao}, when the ontology becomes inconsistent. \ \ The  main contributions in this work are the following:
\begin{enumerate}
\item We introduce a general inconsistency-measure based on an abstract repair-semantics.  We specialize this measure to some well-known classes of repairs: Subset-repairs, most prominently cardinality-repairs, and attribute-based repairs.
\item We introduce answer-set programs to compute the latter inconsistency-measures, and we show that they provide the required expressive power.

\item We obtain data complexity results for the inconsistency measure, showing that its computation (as a decision problem) is NP-complete for denial constraints (DCs) and some classes of functional dependencies.

\item We obtain deterministic and randomized PTIME approximation results for the inconsistency measure, with approximation ratio $d$.
\item We establish that the inconsistency measure behaves well under updates, in that small updates keep the inconsistency measure within narrow boundaries. Furthermore, we establish that the computation of the inconsistency measure  is fixed-parameter tractable when one starts with a consistent instance, and the parameter is the number of updates.
%\item
\end{enumerate}

This paper is structured as follows. Section \ref{sec:prel} reviews background material. Section \ref{sec:incd} introduces a class of abstract, repair-based inconsistency measures. Section \ref{sec:asp} presents and discusses answer-set programs for the computation of the inconsistency measure. Section \ref{sec:comple} presents results on the complexity of the inconsistency measure computation, and some results on its approximate computation. Section \ref{sec:updates} obtains some first results on the behavior of the inconsistency measure under updates. Section \ref{sec:atr} shows how to modify the inconsistency measure in order to make it depend on attribute-based repairs.
Section \ref{sec:open} elaborates on several possible extensions of this work. Appendix A. shows DLV programs for the examples considered in Section \ref{sec:asp}.
\ Material from Section \ref{sec:incd} will appear (and was submitted) as a short communication in \cite{sum18}.

\vspace{-2mm}
\section{Background}\label{sec:prel} \vspace{-2mm}
\subsection{Relational databases and database repairs} \
  A relational schema $\mc{R}$ contains a domain, $\mc{C}$, of constants and a set, $\mc{P}$, of  predicates of finite arities. $\mc{R}$ gives rise to a language $\mf{L}(\mc{R})$ of first-order (FO)  predicate logic with built-in equality, $=$.  Variables are usually denoted by $x, y, z, ...$, and sequences thereof by $\bar{x}, ...$; and constants with $a, b, c, ...$, etc. An {\em atom} is of the form $P(t_1, \ldots, t_n)$, with $n$-ary $P \in \mc{P}$   and $t_1, \ldots, t_n$ {\em terms}, i.e. constants,  or variables.
  An atom is {\em ground} (a.k.a. a tuple) if it contains no variables. A DB  instance, $D$, for $\mc{R}$ is a finite set of ground atoms; and it serves as an  \ignore{The {\em active domain} of a DB instance $D$, denoted ${\it Adom}(D)$, is the set of constants that appear in atoms of $D$.} interpretation structure for  $\mf{L}(\mc{R})$.

A {\em conjunctive query} (\cq) is a \fo \ formula,  $\mc{Q}(\bar{x})$, of the form \ $\exists  \bar{y}\;(P_1(\bar{x}_1)\wedge \dots \wedge P_m(\bar{x}_m))$,
 with $P_i \in \mc{P}$, and (distinct) free variables $\bar{x} := (\bigcup \bar{x}_i) \smallsetminus \bar{y}$. If $\mc{Q}$ has $n$ (free) variables,  $\bar{c} \in \mc{C}^n$ \ is an {\em answer} to $\mc{Q}$ from $D$ if $D \models \mc{Q}[\bar{c}]$, i.e.  $Q[\bar{c}]$ is true in $D$  when the variables in $\bar{x}$ are componentwise replaced by the values in $\bar{c}$. $\mc{Q}(D)$ denotes the set of answers to $\mc{Q}$ from $D$. $\mc{Q}$ is a {\em boolean conjunctive query} (\bcq) when $\bar{x}$ is empty; and when {\em true} in $D$,  $\mc{Q}(D) := \{\nit{true}\}$. Otherwise, it is {\em false}, and $\mc{Q}(D) := \emptyset$. Sometimes CQs are written  in Datalog notation as follows: \ $\mc{Q}(\bar{x}) \leftarrow P_1(\bar{x}_1),\ldots, P_m(\bar{x}_m)$.

In this work we consider integrity constraints (ICs), i.e. sentences of $\mf{L}(\mc{R})$,  that are: (a) {\em denial constraints} \ (DCs), i.e.  of the form $\kappa\!:  \neg \exists \bar{x}(P_1(\bar{x}_1)\wedge \dots \wedge P_m(\bar{x}_m))$,
where $P_i \in \mc{P}$, and $\bar{x} = \bigcup \bar{x}_i$; and (b) {\em functional dependencies} \ (FDs), i.e. of the form  $\varphi\!:  \neg \exists \bar{x} (P(\bar{v},\bar{y}_1,z_1) \wedge P(\bar{v},\bar{y}_2,z_2) \wedge z_1 \neq z_2)$.\footnote{The variables in $\bar{v}$ do not have to go first in the atomic formulas; what matters is keeping the correspondences between the variables in those formulas.} Here,
$\bar{x} = \bar{y}_1 \cup \bar{y}_2 \cup \bar{v} \cup \{z_1, z_2\}$, and $z_1 \neq z_2$ is an abbreviation for $\neg z_1 = z_2$. A {\em key constraint} \ (KC) is a conjunction of FDs: \  $\bigwedge_{j=1}^k \neg \exists \bar{x} (P(\bar{v},\bar{y}_1) \wedge P(\bar{v},\bar{y}_2) \wedge y_1^j \neq y_2^j)$,
with $k = |\bar{y_1}| = |\bar{y}_2|$, and generically $y^j$ stands for the $j$th variable in $\bar{y}$. For example, $\forall x \forall y \forall z(\nit{Emp}(x,y) \wedge \nit{Emp}(x,z) \rightarrow y = z)$, is an FD (and also a KC) that could say that an employee ($x$) can have at most one salary. This FD is usually written as $\nit{EmpName} \rightarrow \nit{EmpSalary}$.
\ In the following, we will include FDs and key constraints among the DCs.
\  If an instance $D$ does not satisfy the set $\Sigma$ of DCs associated to the schema, we say that $D$ is {\em inconsistent}, which is denoted with \ $D \not \models \Sigma$.

When a database instance $D$ does not satisfy its intended ICs, it is {\em repaired}, by deleting or inserting tuples from/into the database. An instance obtained in this way is a {\em repair} of $D$ if it satisfies the ICs and departs in a minimal way from $D$ \cite{Bertossi2011}.  In this work, mainly to fix ideas and simplify the presentation, we consider mostly set $\Sigma$ of ICs that are monotone, in the sense that $D \not \models \Sigma$ and $D \subseteq D'$ imply $D' \not \models \Sigma$. This is the case for DCs.\footnote{Put in different terms, a DC is associated to (or is the negation of) a conjunctive queries $Q$, which is monotone  in the usual sense: \ $D \models Q \mbox{ and } D\subseteq D' \ \Rightarrow \ D' \models Q$.}  For monotone ICs, repairs are obtained by tuple deletions  (later on we will also consider value-updates as repair actions).   \ignore{ DCs are logical formulas of the form $\neg \exists \bar{x}(P_1(\bar{x}_1)\wedge \dots \wedge P_m(\bar{x}_m))$,
where  $\bar{x} = \bigcup \bar{x}_i$; and FDs are of the form  $\neg \exists \bar{x} (P(\bar{v},\bar{y}_1,z_1) \wedge P(\bar{v},\bar{y}_2,z_2) \wedge z_1 \neq z_2)$, with
$\bar{x} = \bar{y}_1 \cup \bar{y}_2 \cup \bar{v} \cup \{z_1, z_2\}$. \ In the following, we will treat FDs as DCs. A database is {\em inconsistent} wrt. a set of ICs $\Sigma$ when $D$ does not satisfy $\Sigma$, denoted $D \not \models \Sigma$.}
 We introduce the most common repairs of databases wrt. DCs  by means of an example.

\begin{example} \label{ex:rep} \  The DB $D = \{P(a), P(e), Q(a,b), R(a,c)\}$ is inconsistent wrt. $\Sigma$ containing the DCs \ $\kappa_1\!: \ \neg \exists x \exists y (P(x) \wedge Q(x,y))$, and \
$\kappa_2\!: \ \neg \exists x \exists y (P(x) \wedge R(x,y))$. Here, \  $D \not \models \{\kappa_1, \kappa_2\}$.

\ignore{
\begin{multicols}{2}
\hspace*{-0.5cm}{\small
\begin{tabular}{c|c|}\hline
$P$&A\\ \hline
&a\\
&e\\ \hhline{~-}
\end{tabular}~~
\begin{tabular}{c|c|c|}\hline
$Q$&A&B\\ \hline
& a & b\\ \hhline{~--}
\end{tabular}~~
\begin{tabular}{c|c|c|}\hline
$R$&A&C\\ \hline
& a & c\\ \hhline{~--}
\end{tabular} }
{\begin{eqnarray*}
\psi_1\!: \ \neg \exists x \exists y (P(x) \wedge Q(x,y)),\\
\psi_2\!: \ \neg \exists x \exists y (P(x) \wedge R(x,y)).
\end{eqnarray*}}
\end{multicols}  }
A {\em subset-repair},  in short {\em S-repair}, of $D$ wrt. $\Sigma$ is a $\subseteq$-maximal subset of $D$ that is consistent, i.e.  no proper superset is consistent. The following are
S-repairs: $D_1 = \{P(e), Q(a,b),$ $ R(a,c)\}$ and $D_2 = \{P(e), P(a)\}$. Under this repair semantics, both repairs are equally acceptable.
\ A {\em cardinality-repair},  in short a {\em C-repair}, is a maximum-cardinality S-repair.  $D_1$  is
the only C-repair. \boxtheorem
\end{example}

\vspace{-2mm}For an instance $D$ and a set $\Sigma$ of DCs, the sets of S-repairs and C-repairs are denoted with $\nit{Srep}(D,\Sigma)$ and $\nit{Crep}(D,\Sigma)$, resp. \ It holds: \ $\nit{Crep}(D,\Sigma) \subseteq \nit{Srep}(D,\Sigma)$.
More generally, for  a set $\Sigma$ of ICs, not necessarily DCs, they can be defined by (cf. \cite{Bertossi2011}): \

\begin{itemize}
\item[(a)]$\nit{Srep}(D,\Sigma) = \{D'~:~ D' \models \Sigma, \mbox{ and } D \bigtriangleup D' \mbox{ is minimal under set inclusion}\}$, and
\item[(b)]
$\nit{Crep}(D,\Sigma) = \{D'~:~ D' \models \Sigma, \mbox{ and } D \bigtriangleup D' \mbox{ is minimal in cardinality}\}$.
\end{itemize}
Here, $D \bigtriangleup D'$ is the symmetric set-difference $(D\smallsetminus D') \cup (D' \smallsetminus D)$.

\vspace{-2mm}
\subsection{Disjunctive answer-set programs}\label{sec:dasps}

We consider answer-set programs (ASPs) \cite{brewka}, and more specifically, disjunctive Datalog programs $\Pi$ with stable model semantics \cite{eiterGottlob97}. They consist of a set $E$ of ground atoms, called the {\em extensional database}, and a  finite number of rules of the form: \
\begin{equation}
A_1(\bar{x}_1) \vee \cdots \vee A_n(\bar{x}_n) \leftarrow P_1(\bar{x}'_1), \ldots, P_m(\bar{x}'_m), \n N_1(\bar{x}''_1), \ldots, \n N_k(\bar{x}''_k),\label{eq:rule}
\end{equation}
with $0\leq n,m,k$, the $A_i, P_j, N_s$ positive atoms, and \ $\cup \bar{x}_i, \cup\bar{x}''_j \subseteq \cup \bar{x}'_s$, i.e. the variables in the $A_i, N_s$ appear all among those
in the $P_j$. The terms in these atoms are constants or variables.

The constants in program $\Pi$ form the (finite) Herbrand universe $U$ of the program. The {\em ground version} of
program $\Pi$, $\nit{gr}(\Pi)$, is obtained by instantiating the variables in $\Pi$ with all possible combinations of
values from $U$. The Herbrand base, $\nit{HB}$, of $\Pi$ consists of all the possible atomic sentences obtained by instantiating the
predicates in $\Pi$ on $U$. A subset $M$ of $\nit{HB}$ is a (Herbrand) model of $\Pi$ if it contains $E$ and satisfies $\nit{gr}(\Pi)$, that is: For every
ground rule $A_1 \vee \ldots \vee A_n \leftarrow P_1, \ldots, P_m, \n N_1, \ldots,
\n N_k$ of $\nit{gr}(\Pi)$, if $\{P_1, \ldots, P_m\} \subseteq M$ and $\{N_1, \ldots, N_k\} \cap M = \emptyset$, then
$\{A_1, \ldots, A_n\} \cap M \neq \emptyset$. $M$ is a {\em minimal model} of $\Pi$ if it is a model of $\Pi$, and no proper subset of $M$ is a model of $\Pi$. $\nit{MM}(\Pi)$ denotes the class of minimal models of $\Pi$.

Now, take $S \subseteq \nit{HB}(\Pi)$, and transform $\nit{gr}(\Pi)$ into a new, positive program $\nit{gr}(\Pi)\!\downarrow \!S$ (i.e. without $\nit{not}$), as follows:
Delete every ground instantiation of a rule (\ref{eq:rule})  for which $\{N_1, \ldots, N_k\} \cap S \neq \emptyset$. Next, transform each remaining ground instantiation of a rule (\ref{eq:rule})  into $A_1 \vee \ldots A_n \leftarrow P_1, \ldots, P_m$. By definition, $S$ is a {\em stable model} of $\Pi$ iff $S \in \nit{MM}(\nit{gr}(\Pi)\!\downarrow \!S)$ \cite{gl91}. A program $\Pi$ may have none, one or several stable models; and each stable model is a minimal model (but not necessarily the other way around) \cite{gelfond}.

\ignore{++++
\subsection{The Repair-Causality Connection}\label{sec:rep-cause}

\vspace{-2mm}
\paragraph{Causes from repairs.} \ In \cite{tocs} it was shown that causes for queries can be obtained from DB repairs.
Consider the BCQ \ ${\mc{Q}\!: \exists \bar{x}(P_1(\bar{x}_1) \wedge \cdots \wedge P_m(\bar{x}_m))}$ that is (possibly unexpectedly) true in  $D$: \ $D \models \mc{Q}$. Actual causes for $\mc{Q}$, their  contingency sets, and responsibilities can be obtained from DB repairs. First,
$\neg \mc{Q}$ is logically equivalent to  the  DC: \vspace{-2mm}
\begin{equation}
{{\kappa(\mc{Q})}\!: \ \neg \exists \bar{x}(P_1(\bar{x}_1) \wedge \cdots \wedge P_m(\bar{x}_m))}. \label{eq:qkappa} \vspace{-1mm}
\end{equation}
So, if $\mc{Q}$ is true in $D$, \ $D$ is inconsistent wrt. $\kappa(\mc{Q})$, giving rise to repairs of $D$ wrt. $\kappa(\mc{Q})$.

Next, we build differences, containing a tuple $\tau$, between $D$ and  S-  or  C-repairs: \vspace{-2mm} \begin{eqnarray}
 \nit{Diff}^s(D,\kappa(\mc{Q}), \tau) \ &=& \ \{ D \smallsetminus D'~|~ D' \in \nit{Srep}(D,\kappa(\mc{Q})), \  \tau \in (D\smallsetminus D')\}, \label{eq:s}\\
 \nit{Diff}^c(D,\kappa(\mc{Q}), \tau) \ &=& \ \{ D \smallsetminus D'~|~ D' \in \nit{Crep}(D,\kappa(\mc{Q})), \ \tau \in (D\smallsetminus D')\}. \label{eq:c}
 \end{eqnarray}

\vspace{-1mm}
It holds \cite{tocs}: \ $\tau \in D$ is an {actual cause} for $\mc{Q}$ iff
$\nit{Diff}^s(D, \kappa(\mc{Q}), \tau) \not = \emptyset$. \ Furthermore, each S-repair $D'$ for which $(D\smallsetminus D') \in \nit{Diff}^s(D, \kappa(\mc{Q}), \tau)$ gives us $(D\smallsetminus (D' \cup \{\tau\}))$ as a subset-minimal contingency set for $\tau$. \ Also, if { $\nit{Diff}^s(D$  $\kappa(\mc{Q}),  \tau) = \emptyset$}, then {$\rho(\tau)=0$}.
 \ Otherwise, { $\rho(\tau)=\frac{1}{|s|}$}, where {  $s \in \nit{Diff}^s(D,$ $\kappa(\mc{Q}), \tau)$} and there is no { $s' \in \nit{Diff}^s(D,\kappa(\mc{Q}), \tau)$} with { $|s'| < |s|$}.
\ As a consequence we obtain that $\tau$ is a most responsible actual cause  for $\mc{Q}$ \ iff \
$\nit{Diff}^c\!(D,\kappa(\mc{Q}), \tau) \not = \emptyset$.

\vspace{-2mm}
\begin{example} (ex. \ref{ex:cause} cont.) \label{ex:kappa} \  With the same instance $D$ and query $\mc{Q}$, we consider the
DC \ $\kappa(\mc{Q})$:  \ $\neg \exists x\exists y( S(x)\wedge R(x, y)\wedge S(y))$, which is not satisfied by $D$.
\ Here, ${\nit{Srep}(D, \kappa(\mc{Q})) =\{D_1, D_2,D_3\}}$ and ${\nit{Crep}(D, \kappa(\mc{Q}))=\{D_1\}}$, with
$D_1=$ $ \{R(a_4,a_3),$ $ R(a_2,a_1), R(a_3,a_3), S(a_4), S(a_2)\}$, \  $D_2 = \{ R(a_2,a_1), S(a_4),S(a_2),$  $S(a_3)\}$, \  $D_3 =$ $\{R(a_4,a_3), R(a_2,a_1), S(a_2),S(a_3)\}$.

For tuple \ ${R(a_4,a_3)}$,  \ ${\nit{Diff}^s(D, \kappa(\mc{Q}), {R(a_4,a_3)})=\{D \smallsetminus D_2\}}$ $= \{ \{ R(a_4,a_3),$ \linebreak $ R(a_3,a_3)\} \}$. So,
 $R(a_4,a_3)$ is an actual cause,  with responsibility $\frac{1}{2}$. \ Similarly, $R(a_3,a_3)$ is an actual cause, with responsibility $\frac{1}{2}$.
\ For tuple ${S(a_3)}$,  \  $\nit{Diff}^c(D, \kappa(\mc{Q}), S(a_3)) =$ $ \{D \smallsetminus D_1\} =\{ S(a_3) \}$.
So, $S(a_3)$
is an actual cause,  with responsibility 1, i.e. a  {most responsible cause}. \boxtheorem
\end{example}

\vspace{-2mm}
It is also possible, the other way around, to characterize repairs in terms of causes and their contingency sets. Actually this connection can be used to obtain complexity results for
causality problems from repair-related computational problems \cite{tocs}. Most computational problems related to repairs, specially C-repairs, which are related to most responsible causes, are provably hard.
This is reflected in a high complexity for responsibility \cite{tocs} \ (cf. Section \ref{sec:compl}).
++++}

\vspace{-2mm}
\section{Repair Semantics and Inconsistency Degrees}\label{sec:incd}\vspace{-2mm} \ In general terms, a {\em repair semantics} {\sf S} for  a schema $\mc{R}$ that includes a set $\Sigma$ of ICs assigns to each  instance $D$ for $\mc{R}$  (which may not satisfy $\Sigma$),  a class $\nit{Rep}^{\sf S}(D,\Sigma)$ of
{\sf S}{\em -repairs} of $D$ wrt. $\Sigma$, which are instances of $\mc{R}$ that satisfy $\Sigma$ and depart from $D$ according to some minimization criterion.
\ Several repair semantics have been considered in the literature, among them and beside those introduced in Example \ref{ex:rep}, {\em prioritized repairs} \cite{stawo},  and {\em attribute-based repairs} that change attribute values by other data values, or by a null value, {\sf NULL}, as in SQL databases \cite{foiks18} (cf. Section \ref{sec:atr}).

According to our  take on how a database inconsistency degree depends on database repairs, we define the {\em inconsistency degree} of an instance $D$ wrt. a set of ICs $\Sigma$ in relation to a given repair semantics {\sf S}, as the  distance from $D$ to the class $\srep$:
\begin{equation}
\mbox{\nit{inc-deg}}^{\sf S}(D,\Sigma) := \nit{dist}(D,\srep). \label{eq:dist}
\end{equation}

\vspace{-1mm}This is an abstract measure that depends on {\sf S} and a given function that returns the distance, $\nit{dist}(W,\mc{W})$, from a world $W$ to a set $\mc{W}$ of possible worlds, which in this case are database instances. Under the assumption that any repair semantics should return $D$ when $D$ is consistent wrt. $\Sigma$ and
$\nit{dist}(D,\{D\}) = 0$, a consistent instance $D$ should have $0$ as inconsistency degree.\footnote{Abstract distances between two point-sets are investigated in \cite{eiterMannila}, with their computational properties. Our setting is a particular case.}

Notice that  the class $\srep$ might contain instances that are not sub-instances of $D$, for example, for different forms of {\em inclusion dependencies} (INDs) we may want to insert tuples;\footnote{For INDs repairs based only on tuple deletions can be considered \cite{chomicki}.} or  even under DCs, we may want to appeal to  attribute-based repairs. \ignore{For example, \cite{wijsen} investigates repairs of this kind; attribute values can be changed by other values in the data domain. In \cite{tkde,tplp} replacement on values by a null \`a la SQL (or at least that disallows joins and comparisons through it) are proposed and investigated, and similarly in \cite[sec. 7.4]{tocs}, to capture attribute-level causes.} {\em In the following, until further notice,  we consider only repairs that are sub-instances of the given instance.} \ Still this leaves much room open for different kinds of repairs. For example, we may prefer to delete some tuples over others \cite{stawo}. Or, as in database causality \cite{suciu,tocs}, the database can be partitioned into {\em endogenous} and {\em exogenous} tuples, assuming we have more control on the former, or we trust more the latter; and we prefer {\em endogenous repairs} that delete only, or preferably, endogenous tuples \cite{foiks18} (cf. Example \ref{ex:endo} below).

\vspace{-2mm}
\subsection{An inconsistency measure}\label{sec:ours}\vspace{-2mm} Here we consider a concrete instantiation of $\mbox{\nit{inc-deg}}^{\sf S}(D,\Sigma)$ in (\ref{eq:dist}), and to fix ideas, only DCs. For them,  the repair semantics $\nit{Srep}(D,\Sigma)$ and $\nit{Crep}(D,\Sigma)$ are  particular cases of repair semantics
 {\sf S} where each $D' \in \srep$ is maximally contained in $D$. On this basis, we can define: \vspace{-3mm}
\begin{eqnarray}
\hspace*{-1mm}\mbox{\nit{inc-deg}}^{{\sf S},g_3\!}(D,\Sigma)  &:=&  \nit{dist}^{g_3\!}(D,\srep)  :=  \frac{|D| \! - \! \nit{max}\{ |D'| : D' \in \srep  \}}{|D|} \nonumber\\
&=& \frac{ \nit{min} \{|D \smallsetminus D'|~:~ D' \in \srep  \}}{|D|}, \hspace*{-1mm}\label{eq:distG3}
\end{eqnarray}

 \vspace{-3mm} \noindent inspired by distance $g_3$ in \cite{mannila} to measure the degree of violation of an FD by a database.\footnote{Other possible measures for single FDs and relationships between them can be found in \cite{mannila}.} This measure can be applied more generally as a ``quality measure", not only in relation to inconsistency, but also whenever  possibly several intended ``quality versions" of a dirty database exist, e.g. as determined by additional contextual information \cite{context}.

 Particularly prominent are the instantiation of (\ref{eq:distG3}) on the S-repair and C-repair semantics:
 \begin{eqnarray}
\mbox{\nit{inc-deg}}^{s,g_3\!}(D,\Sigma) &:=& \frac{|D| - \nit{max}\{ |D'| ~:~D' \in \nit{Srep}(D,\Sigma)  \}}{|D|}\label{eq:s}\\
\mbox{\nit{inc-deg}}^{c,g_3\!}(D,\Sigma) &:=& \frac{|D| - \nit{max}\{ |D'| ~:~D' \in  \nit{Crep}(D,\Sigma) \}}{|D|}\label{eq:c}
\end{eqnarray}

\begin{example} (ex. \ref{ex:rep} cont.) \label{ex:rep2} \ignore{Consider again
$D = \{P(a), P(e), Q(a,b), R(a,c)\}$, which violates the set of DCs $\Sigma = \{\kappa_1, \kappa_2\}$.} Here, $\nit{Srep}(D,\Sigma) = \{D_1, D_2 \}$, and
$\nit{Crep}(D,\Sigma) = \{D_1 \}$\ignore{, with $D_1 = \{P(e),$ $ Q(a,b), R(a,c)\}$ and $D_2 = \{P(a), P(e)\}$}. They provide the inconsistency degrees:
$$\mbox{\nit{inc-deg}}^{s,g_3\!}(D,\Sigma)  = \frac{4 -|D_1|}{4} = \frac{1}{4}, \ \mbox{ and } \
\mbox{\nit{inc-deg}}^{c,g_3\!}(D,\Sigma) =  \frac{4 - |D_1|}{4} = \frac{1}{4},$$
 respectively. \boxtheorem
\end{example}

It holds  $\nit{Crep}(D,\Sigma) \subseteq \nit{Srep}(D,\Sigma)$, but $\nit{max}\{|D'|~:~D' \in \nit{Crep}(D,\Sigma)\}$ $ = \nit{max}\{|D'|~:~D' \in \nit{Srep}(D,\Sigma)\}$,  so it holds $\mbox{\nit{inc-deg}}^{s,g_3\!}(D,\Sigma) = \mbox{\nit{inc-deg}}^{c,g_3\!}(D,\Sigma)$.
 \ This measure always takes a value between $0$ and $1$. The former when $D$ is consistent (so it itself is its only repair).

 The measure takes the value $1$ only when $\srep = \emptyset$ \ (assuming that
 $\nit{max} \{$ $|D'| ~:$ $~ D' \in \emptyset\} = 0$), i.e. the database is {\em irreparable}, which is never the case for DCs and S-repairs: there is always an S-repair. However, it could be irreparable with different, but related repair semantics.
  For example, as mentioned above, in database causality \cite{suciu} tuples can be endogenous or exogenous, being the former those we can play with, e.g. applying virtual updates on them, producing counterfactual scenarios. On this basis, one can define {\em endogenous repairs}, which are obtained by updating only endogenous tuples \cite{tocs}.

\begin{example} (ex. \ref{ex:rep2} cont.) \label{ex:endo} Assume $D$ is partitioned into endogenous and exogenous tuples, say resp. \ $D = D^n \stackrel{.}{\cup} D^x$, with $D^n = \{ Q(a,b), R(a,c) \}$ and $D^x = \{ P(a), P(e)\}$. In this case, the {\em endogenous-repair semantics} that allows only a minimum number of deletions of endogenous tuples, defines the class of repairs: $\nit{Crep}^n(D,\Sigma) = \{D_2\}$, with $D_2$ as above. In this case,\footnote{For certain forms of {\em prioritized repairs}, such as endogenous repairs, the normalization coefficient $|D|$ might be unnecessarily large. In this particular case, it might be better to use $|D^n|$.} \ $\mbox{\nit{inc-deg}}^{c,n,g_3\!}(D,\Sigma) = \frac{4-2}{4} = \frac{1}{2}$.
\ Similarly, if now $D^n = \{ P(a), Q(a,b)  \}$ and $D^x = \{P(e), R(a,c)\}$, there are no endogenous repairs, and \ $\mbox{\nit{inc-deg}}^{c,n,g_3\!}(D,\Sigma) = 1$.
\boxtheorem
\end{example}

%XXX  Properties XXX

\vspace{-2mm}
\section{ASP-Based Computation of the Inconsistency Measure} \label{sec:asp}\vspace{-2mm} We concentrate here on measure $\mbox{\nit{inc-deg}}^{c,g_3\!}(D,\Sigma)$ \ in (\ref{eq:c}); and  more generally, on $\mbox{\nit{inc-deg}}^{s,g_3\!}(D,\Sigma)$, which can be computed through the maximum cardinality of an S-repair for $D$ wrt. $\Sigma$, or, equivalently, using the cardinality of a (actually, every) repair in $\nit{Crep}(D,\Sigma)$. \ This can be done  through a compact specification of repairs by means of ASPs.\footnote{This approach was followed in \cite{foiks18} to compute maximum {\em responsibility degrees} of database tuples as causes for violations of DCs, appealing to a causality-repair connection \cite{tocs}.} More precisely, given a database instance $D$ and a set of ICs $\Sigma$ (not necessarily DCs), it is possible to write an ASP whose intended models, i.e. the {\em stable models} or {\em answer sets}, are in one-to-one correspondence with the S-repairs of $D$ wrt. $\Sigma$. Cf.  \cite{monica} for a general formulation. Here we show only some cases of ICs and examples. In them we use, only to ease the formulation and presentation, global unique tuple identifiers (tids), i.e. every tuple $R(\bar{c})$ in $D$ is represented as $R(t;\bar{c})$ for some integer (or constant) $t$  that is not used by any other tuple in $D$.

If $\Sigma$ is a set of DCs containing  \ $\kappa\!:  \neg \exists \bar{x}(P_1(\bar{x}_1)\wedge \dots \wedge P_m(\bar{x}_m))$, we first introduce for a predicate $P_i$ of the database schema, a nickname predicate $P_i'$  that has, in addition to a first attribute for tids,  an extra, final attribute to hold an annotation from the set $\{\sf{d}, \sf{s}\}$, for ``delete" and ``stays", resp. \ Nickname predicates are used to represent and compute repairs. Next, the {\em repair-ASP}, $\Pi(D,\Sigma)$, for $D$ and $\Sigma$ contains all the tuples in $D$ as facts (with tids), plus the following rules for $\kappa$:
\begin{eqnarray*}
P_1'(t_1;\bar{x}_1,\sfd)\vee \cdots \vee P_m'(t_n;\bar{x}_m,\sfd) &\leftarrow& P_1(t_1;\bar{x}_1), \dots, P_m(t_m;\bar{x}_m). \\
P_i'(t_i;\bar{x}_i,\sfs) &\leftarrow& P_i(t_i;\bar{x}_i), \ \nit{not} \ P_i'(t_i;\bar{x}_i,\sfd). \ \ \ \  \ i=1,\cdots,m.
\end{eqnarray*}
A stable   model $M$ of the program determines a repair $D'$ of $D$: \ $D' := \{P(\bar{c})~|$ \linebreak $P'(t;\bar{c},\sfs) \in M\}$, and every repair can be obtained in this way \cite{monica,tplpP2P}.

For an FD in $\Sigma$, say $\varphi\!: \ \neg \exists xyz_1z_2vw(R(x,y,z_1,v) \wedge R(x,y,z_2,w) \wedge z_1 \neq z_2)$, which makes the third attribute functionally depend upon the first two, the repair program contains the rules:
\begin{eqnarray*}
R'(t_1;x,y,z_1,v,\sfd) \vee R'(t_2;x,y,z_2,w,\sfd) &\leftarrow& R(t_1;x,y,z_1,v), R(t_2;x,y,z_2,w),\\ && \hspace*{4.4cm}z_1 \neq z_2.\\
R'(t;x,y,z,v,\sfs) &\leftarrow& R(t;x,y,z,v), \ \nit{not} \ R'(t;x,y,z,v,\sfd).
\end{eqnarray*}
 For DCs and FDs, the repair programs can be made {\em normal}, i.e.  non-disjunctive, by moving all the disjuncts but one, in turns, in negated form to the body of the rule \cite{monica} (cf. Section \ref{sec:dlv}). For example, the rule
$P(a) \vee R(b) \leftarrow \nit{Body}$, can be written as the two rules \ $P(a)  \leftarrow \nit{Body}, \nit{not} \ R(b)$ and $R(b) \leftarrow \nit{Body}, \nit{not} \ P(a)$.\footnote{This transformation preserves the semantics, because these repair-ASPs turn out to be head-cycle-free \cite{monica}.} Still the resulting program can be {\em non-stratified}
if there is recursion via negation \cite{gelfond}, as in the case of FDs, and  DCs with self-joins.

\begin{example} (ex. \ref{ex:rep} cont.)\label{ex:rep2} The initial instance with tids is $D=\{P(1,e), P(2,a),$ $ Q(3,a,b),$  $R(4,a,c), \}$. The repair program contains the following rules, with the first and second for $\kappa_1$ and $\kappa_2$, resp.:
\begin{eqnarray*}
P'(t_1;x,\sfd) \vee Q'(t_2;x,y,\sfd) &\leftarrow& P(t_1;x), Q(t_2;x, y). \\%\label{eq:repRule}\\
P'(t_1;x,\sfd) \vee R'(t_2;x,y,\sfd) &\leftarrow& P(t_1;x), R(t_2;x, y).\\% \label{eq:repRule2}\\
P'(t;x,\sfs) &\leftarrow& P(t;x), \ \nit{not} \ P'(t;x,\sfd). \ \ \ \ \mbox{ etc. }% \nonumber
\end{eqnarray*}
The {\em repair program} $\Pi(D,\{\kappa_1,\kappa_2\})$ has the stable models: \ $\mc{M}_1 = \{P'(1,e,{\sf s}),$\linebreak $Q'(3,a,b,{\sf s}),$  $ R'(4,a,c,{\sf s}),$ $ P'(2,a,{\sf d}) \}$ $ \cup \ D$ \ and \ $\mc{M}_2 = \{P'(1,e,{\sf s}), P'(2,a,{\sf s}),$ $Q'(3,a,b,{\sf d}), R'(4,a,c,{\sf d}) \} \cup \ D$, which  correspond to  the S-repairs $D_1, D_2$, resp.
\boxtheorem \end{example}

Similar repair programs can be produced to specify {\em attribute-based repairs} that, instead of deleting (or inserting) tuples, change attribute values in existing tuples. This is the case, for example, when one allows changing values into a null value as in SQL databases, on the assumption that joins and comparisons through nulls do not hold \cite{foiks18}. This becomes relevant in Section \ref{sec:atr}.

Now, and back to tuple-based repairs, to compute $\mbox{\nit{inc-deg}}^{c,g_3\!}(D,\Sigma)$, for the C-repair semantics, we can add rules to $\Pi$ to collect the {\em tids} of tuples deleted from the database, a rule with aggregation to compute the number of deleted tuples, plus a {\em weak program-constraint} \cite{dlv} that eliminates all the stable models (equivalently, S-repairs) that violate the constraint a non-minimum number of times:
\begin{eqnarray*}
\nit{Del}(t) &\leftarrow& P_i'(t,\bar{x}_i,{\sf d}). \ \ \ \ \ \  i = 1, \ldots, m\\
\nit{NumDel}(n) &\leftarrow& \# \nit{count}\{t : \nit{Del}(t)\} = n.\\
&:\sim& \nit{Del}(t).
\end{eqnarray*}
In each model of the program, the first rules collect the tids of deleted tuples, and the second rule counts the total number of deletions. The last rule keeps only the models where the number of deletions is a minimum.\footnote{If we had a (hard) program-constraint instead, written \ $\leftarrow \nit{Del}(t)$, we would be prohibiting the satisfaction of the rule body (in this case, deletions would be prohibited), and we would be keeping only the models where there are no deletions. This  would return no model or the original $D$ depending on whether $D$ is inconsistent or not.} The reason for introducing weak constraints is that, without them, the stable models of the program capture the S-repairs, i.e. $\subseteq$-maximal and consistent sub-instances of $D$, but not necessarily maximum in cardinality. With the weak constraint we keep only cardinality repairs.

 \begin{example} (ex. \ref{ex:rep2} cont.) \label{ex:del} If we add to $\Pi$ the rule \
\ $\nit{Del}(t) \leftarrow R'(t,x,y,{\sf d})$, and similarly for $Q'$ and $P'$;  \ and next, a rule to count the deleted tuples,
$\nit{NumDel}(n) \leftarrow \# \nit{count}\{t : \nit{Del}(t)\} = n$, the stable model $\mc{M}_1$ of the original program would be extended with the atoms $\nit{Del}(2), \nit{NumDel}(1)$. Similarly for $\mc{M}_2$.

 If we also add the weak constraint \ $:\sim \nit{Del}(t)$, only (the extended) model $\mc{M}_1$ remains. It corresponds to the only C-repair. \boxtheorem
 \end{example}

 The value for $\nit{NumDel}$ in any of the remaining models can be used to compute  $\mbox{\nit{inc-deg}}^{c,g_3\!}(D,\Sigma)$. So, there is no need to explicitly compute all stable models, their sizes, and compare them. This value can be obtained by means of the query \linebreak ``$:\!\!- \ \nit{NumDel}(x)?$", answered by the extended program  under the {\em brave semantics} (returning answers that hold in {\em some} of the stable models). \ Appendix A. shows an extended example that uses DLV-Complex \cite{dlv,calimeri08} for the computation with the ASPs we introduced in this section.

It has been established that brave reasoning with repair programs for DCs with weak constraints is $\Delta^P_2(\nit{log}(n))$-complete in data complexity, i.e. in the size of the database \cite{monica,buca}. As we will see in Section
 \ref{sec:comple} (cf. Theorem \ref{thm:fp}), this complexity matches the intrinsic complexity of the computation of the inconsistency measure.

\vspace{-2mm}
\section{Complexity of the Inconsistency Measure Computation}\label{sec:comple}
\vspace{-2mm}
We recall first that the {\em functional complexity class} $\nit{FP}^{\nit{NP(log(n))}}$ contains computation problems whose counterparts as decision problems
are in the class $\nit{P}^{\nit{NP(log(n))}}$, i.e. they are solvable in polynomial time with a logarithmic number of calls to an $\nit{NP}$-oracle \cite{papa}.

\begin{theorem} \label{thm:fp} \em For DCs, computing $\mbox{\nit{inc-deg}}^{c,g_3\!}(D,\Sigma)$ belongs to the functional class $\nit{FP}^{\nit{NP(log(n))}}$; and there is a relational schema and a set of DCs $\Sigma$, such that
computing $\mbox{\nit{inc-deg}}^{c,g_3\!}(D,\Sigma)$  is $\nit{FP}^{\nit{NP(log(n))}}$-complete (all this in data complexity, i.e. in the size of $D$).\boxtheorem
\end{theorem}

 This result and the complexity of ASP evaluation (cf. last paragraph of Section \ref{sec:asp}) show that the normal ASPs introduced in Section \ref{sec:asp} have the right expressive power to deal with the computational problem at hand.
\ We wonder whether we obtain a similar result for FDs. Although for the inconsistency measure the difference between S- and C-repairs does not matter, the next example shows first that there is a difference   between S- and C-repairs in the presence of FDs.

 \begin{example}\label{ex:fds} Consider the schema $R(A,B,C)$, with $\Sigma$ containing the FDs \ $A \rightarrow B$ and $C \rightarrow B$, and the inconsistent instance $D = \{R(a,b,d), R(a,e,c), R(a,b,c)\}$. The S-repairs are $D_1 = \{R(a,b,d), R(a,b,c)\}$ and $D_2 = \{R(a,e,c)\}$. The only C-repair is $D_1$, and $\mbox{\nit{inc-deg}}^{c,g_3\!}(D,\Sigma) = \frac{1}{3}$. \boxtheorem
 \end{example}

%\begin{remarkLB}
\begin{remark}
\label{rem:cg}  In the following we make use several times of the fact that, for a set $\Sigma$ of DCs and an instance $D$, one can build a {\em conflict-hypergraph}, $\nit{CG}(D,\Sigma)$, whose vertices are the tuples in $D$ and hyperedges are subset-minimal sets of tuples that simultaneously participate in the violation of one of the DCs in $\Sigma$ \cite{chomicki,icdt07}. More precisely, for a DC $\kappa\!: \ \neg \exists \bar{x}(P_1(\bar{x}_1) \wedge \ldots \wedge P_l(\bar{x}_l))$ in $\Sigma$, $S \subseteq D$ forms a hyperedge, if $S$ satisfies the BCQ associated to $\kappa$,  $\mc{Q}^\kappa \leftarrow P_1(\bar{x}_1), \ldots, P_l(\bar{x}_l)$, and $S$ is subset-minimal for this property.\footnote{More technically, each DC $\kappa\!: \ \neg \exists \bar{x}(P_1(\bar{x}_1) \wedge \ldots \wedge P_l(\bar{x}_l) \wedge \ldots)$ gives rise to  conjunctive queries $\mc{Q}^\kappa_{P_l}(\bar{x}_l) \leftarrow P_1(\bar{x}_1), \ldots, P_l(\bar{x}_l), \ldots$. A tuple $P(\bar{a})$ participates in the violation of $\kappa$ if $\bar{a}$ is an answer
to $\mc{Q}^\kappa_{P}(\bar{x})$.} A C-repair turns out to be the complement of a minimum-size vertex cover for the conflict-hypergraph; equivalently, of a minimum-size hitting-set for the set of hyperedges;  or, equivalently, a maximum-size independent set of $\nit{CG}(D,\Sigma)$.\boxtheorem
\end{remark}
%\end{remarkLB}

Towards establishing  that Theorem \ref{thm:fp} still holds for FDs, we first observe:

 \begin{lemma} \label{lemma:max} \em There is a fixed relational schema and a set of FDs $\Sigma$, such that verifying for an instance $D$ if the conflict-graph $\nit{CG}(D,\Sigma)$ has an independent set of size $k$ is \nit{NP}-complete in the size of $D$.\boxtheorem
 \end{lemma}

 \begin{corollary} \label{cor:fds} \em There is a fixed relational schema and a set of FDs $\Sigma$, such that verifying for a database instance $D$ if it has a C-repair of size at least $k$ is \nit{NP}-complete in the size of $D$.\boxtheorem
 \end{corollary}

 \begin{theorem} \label{thm:fixed}\em
 There is a fixed relational schema and a set $\Sigma$ of two FDs, such that computing $\mbox{\nit{inc-deg}}^{c,g_3\!}(D,\Sigma)$  is  $\nit{FP}^{\nit{NP(log(n))}}$-complete in data complexity.\boxtheorem
 \end{theorem}

 From this result we obtain that computing  the $\mbox{\nit{inc-deg}}^{s,g_3\!}$ measure is $\nit{FP}^{\nit{NP(log(n))}}$-complete in data complexity.  As claimed in \cite[page 132]{mannila}, it can be computed in $O(\nit{sort}(R(D)))$ for a single FD, where $\nit{sort}(R(D)$ is the time it takes to sort relation $R$ in $D$. However, as  Theorem \ref{thm:fixed} states, the complexity can be higher already for two FDs. It is interesting to highlight that in \cite{benny} it is established that if a set of FDs is ``simplifiable", then a C-repair can be computed in polynomial time. Clearly if we can build such a repair, we can immediately compute the inconsistency measure (one C-repair suffices), and in polynomial time. As expected, the set of FDs in Theorem \ref{thm:fixed}, being of the form $\{A \rightarrow B, \ B \rightarrow C\}$ is not simplifiable.  

  Despite the high-complexity results above,
 there is a good polynomial-time algorithm, $\mbox{\nit{appID}}$, that  approximates $\mbox{\nit{inc-deg}}^{c,g_3\!}(D,\Sigma)$.

 \begin{theorem} \label{thm:det}\em
 There is a polynomial-time, deterministic algorithm that returns $\mbox{\nit{appID}}(D,$ $\Sigma)$, an approximation to $\mbox{\nit{inc-deg}}^{c,g_3\!}(D,\Sigma)$, within the constant factor $d$ that is the maximum number of atoms in a DC in $\Sigma$, i.e. \ $\mbox{\nit{appID}}(D,\Sigma)  \leq d \times  \mbox{\nit{inc-deg}}^{c,g_3\!}(D,\Sigma)$.\boxtheorem
 \end{theorem}

Since for FDs conflict hypergraphs become conflict graphs, we immediately  obtain:

\begin{corollary}
\em For $\Sigma$ a set of FDs, $\mbox{\nit{appID}}(D,\Sigma)$ is a polynomial-time 2-approximation for $\mbox{\nit{inc-deg}}^{c,g_3\!}(D,\Sigma)$, i.e. \ $\mbox{\nit{appID}}(D,\Sigma) \leq 2 \times \mbox{\nit{inc-deg}}^{c,g_3\!}(D,\Sigma)$. \boxtheorem
 \end{corollary}

Another approach to the approximate computation of the inconsistency measure is based on randomization applied to a relaxed, linear-programming version of the hitting-set (HS) problem  for the set of $d$-bounded hyperedges  (or, equivalently, as vertex-covers in hypergraphs with $d$-bounded hyperedges). In our case, this occurs when each of the DCs in $\Sigma$ has a number of atoms bounded by $d$. In this case, we say $\Sigma$ is $d$-bounded, and the hyperedges in the conflict-hypergraph have all size at most $d$. \
The algorithm in \cite{random} returns a ``small", possibly non-minimum HS, which in our case is a set of database tuples whose removal from $D$ restores consistency. The size of this HS approximates the numerator of the inconsistency measure.

\begin{proposition} \label{prop:random} \em
There is a polynomial-time,  randomized algorithm that approximates $\mbox{\nit{inc-deg}}^{c,g_3\!}(D,\Sigma)$ within a $d$-ratio and with probability
$\frac{3}{5}$. \boxtheorem
\end{proposition}

Notice that $d$ in this result is determined by the fixed
set of DCs, and does not depend on $D$. \
Actually, as shown in \cite{random}, the ratio of the algorithm can be improved to $(d - \frac{8}{\Delta})$, where $\Delta \leq \frac{1}{4}|D|^{\frac{1}{4}}$ is the maximum degree of a vertex, i.e. in our case the maximum number of tuples that co-violate a DC (possibly in company of other tuples) with any fixed tuple.\footnote{It is known that there is no polynomial-time approximation with ratio of the form $(d - \epsilon)$ for any constant $\epsilon$ \cite{khot}.} As above, for conflict-graphs associated for example to FDs, $d =2$.

\vspace{-2mm}
\section{Inconsistency Degree under Updates}\label{sec:updates}\vspace{-2mm}

Let us assume we have a $\mbox{\nit{inc-deg}}^{s,g_3}(D,\Sigma)$ for an instance $D$ and a set of DCs $\Sigma$. If, possibly virtually or hypothetically for exploration purposes, we  insert $m$ new tuples into $D$, the resulting instance, $D'$, may suffer from more IC violations than $D$. The question is how much can the inconsistency measure change. The next results tell us that there are no unexpected jumps in inconsistency degree. They can be seen as reflecting {\em continuity} properties of the inconsistency measure.

\begin{proposition}\label{prop:cont1} \em
Given an instance $D$ and a set $\Sigma$ of DCs, if $\epsilon \times |D|$ new tuples are added to $D$, with $0 < \epsilon <1$, obtaining instance $D'$, then  $\mbox{\nit{inc-deg}}^{c,g_3}(D',\Sigma) \leq \mbox{\nit{inc-deg}}^{c,g_3}(D,\Sigma) + \frac{1}{1 + \frac{1}{\epsilon}}$. \ Furthermore, \ $\mbox{\nit{inc-deg}}^{c,g_3}(D,\Sigma) \leq \frac{1}{1-\epsilon} \times \mbox{\nit{inc-deg}}^{c,g_3}(D',\Sigma)$.\boxtheorem
\end{proposition}

When tuples are deleted, the number of DC violations can only decrease, but also the reference size of the database decreases. However, the inconsistency degree stays within a tight upper bound.
\begin{proposition} \label{prop:cont2} \em
Given an instance $D$ and a set $\Sigma$ of DCs, if $\epsilon \times |D|$ tuples are deleted from $D$, with $0 < \epsilon <1$, obtaining instance $D'$, then  $\mbox{\nit{inc-deg}}^{c,g_3}(D',\Sigma) \leq \frac{1}{1-\epsilon} \times \mbox{\nit{inc-deg}}^{c,g_3}(D,\Sigma)$. \ Furthermore, $\mbox{\nit{inc-deg}}^{c,g_3}(D,\Sigma) \leq \frac{1}{1-\epsilon} \times \mbox{\nit{inc-deg}}^{c,g_3}(D',\Sigma) + \epsilon$;  and the last term can be eliminated if the deleted tuples did not participate in DC violations in $D$.\boxtheorem
\end{proposition}

A natural situation occurs when one has a fully consistent database $D$ wrt. a set $\Sigma$ of DCs, and one adds a set $U$ of $m$ tuples (deletions will not affect consistency). The question is about the cost of computing the inconsistency measure. Actually, it turns out that if $\Sigma$ is $d$-bounded, then computing the inconsistency measure is fixed-parameter tractable
 \cite{flum}, where the fixed parameter is $m$.

 \begin{theorem} \label{theo:fpt} \em For a fixed set of DCs $\Sigma$ that is bounded by $d$, a database $D$ that is consistent wrt. $\Sigma$, and $U$ a set of extra tuples, computing $\mbox{\nit{inc-deg}}^{c,g_3\!}(D \cup U,\Sigma)$ is {\em fixed-parameter tractable} with  parameter $m = |U|$. More precisely, there is an algorithm that computes the inconsistency measure in time $O(\nit{log}(m) \times (C^m + m N))$, where $N =|D|$, $m = |U|$, and $C$ is a constant that depends on $d$. \boxtheorem
 \end{theorem}

The complexity is exponential in the number of updates, but linear in the size of the initial database. In many situations, $m$ would be relatively small in comparison to $|D|$.
 \ In Section \ref{sec:incre} we further discuss the incremental approximate computation of the inconsistency measure.

\vspace{-2mm}
\section{Adapting $\mbox{\nit{inc-deg}}^{s,g_3}$ to attribute-based repairs}\label{sec:atr}
\vspace{-2mm}

  Database repairs that are based on changes of attribute values in tuples have been considered in \cite{wijsen,IS2008}, and implicitly in \cite{tkde}. We rely here on repairs introduced in \cite{foiks18}, which we briefly present by means of an example. (We believe the developments in this section could be applied to inconsistency measures based on repairs that update attribute values using other constants from the domain \cite{wijsen,IS2008}.)

\begin{example}  \label{ex:cause2}  For the database instance $D = \{S(a_2), S(a_3), R(a_3,a_1), R(a_3,a_4),$\linebreak $ R(a_3,a_5)\}$,  and the DC \
$\kappa: \ \neg \exists x \exists y ( S(x) \land R(x, y))$,  \ it holds \ ${D \not \models \kappa}$. \
Notice that value $a_3$  matters here in that it enables the join,   e.g. $D \models S(a_3) \wedge R(a_3,a_1)$, which could be avoided by replacing it by a null value as used in SQL databases.

More precisely, for the instance $D_1 = \{S(a_2), S(a_3), R(\nn,a_1),$ $R(\nn,a_4),$  $ R(\nn,a_5)\}$, where $\nn$ stands for the null value, which cannot be used to satisfy a join, it holds \ $D_1 \models \kappa$. \ Similarly with $D_2 = \{S(a_2), S(\nn),$$ R(a_3,a_1),$ $ R(a_3,a_4),$ $ R(a_3,a_5)\}$, and $D_3 = \{S(a_2), S(\nn), R(\nn,a_1),$ $R(\nn,a_4),$ $ R(\nn,a_5)\}$, among others obtained from $D$ through replacement of attribute values by \nn.
\boxtheorem
\end{example}

In relation to the special constant $\nn$ we assume that all atoms with built-in comparisons, say $\nn \ \theta \ \nn$, and $\nn \ \theta \ c$, with $c$ a non-null constant, are all false for $\theta \in \{=,\neq, <, >, \ldots\}$. In particular, since a join, say  $R(\ldots, x) \wedge S(x,\ldots)$, can be written as $R(\ldots, x) \wedge S(x',\ldots) \wedge x=x'$, it can never be satisfied through \nn. This assumption is compatible with the use of {\footnotesize {\sf NULL}} in SQL databases (cf. \cite[sec. 4]{tplpP2P} for a detailed discussion, also \cite[sec. 2]{tkde}). \ Changes of
attribute
values by  \nit{null} as repair actions offer a natural and deterministic solution that appeals to {\em the} generic data value used in SQL databases to reflect the uncertainty and incompleteness in/of the database that inconsistency produces. \ In order to keep track of changes, we introduce numbers as first
arguments in tuples, as  global, unique tuple identifiers (tids).

\begin{example}  \label{ex:cause3} (ex. \ref{ex:cause2} cont.) \ With tids $D$  becomes  $D = \{S(1;a_2), S(2;a_3), R(3;a_3,a_1), \linebreak R(4;a_3,a_4),$ $ R(5;a_3,a_5)\}$; and $D_1$ becomes $D_1 = \{S(1;a_2), S(2;a_3), R(3;\nn,a_1),$ $R(4;\nn,a_4),$ $ R(5;\nn,a_5)\}$. \ The changes are collected in $\Delta^\nn(D,D_1) := \{R[3;1],$ $ R[4;1],R[5;1]\}$, showing that (the original) tuple (with tid) $3$ has its first-argument changed into $\nn$, etc. \ Similarly, $\Delta^\nn(D,D_2) := \{S[2;1]\}$, and
$\Delta^\nn(D,D_3) := \{S[2;1], R[3;1],$ $ R[4;1],R[5;1]\}$.

$D_1$ and $D_2$ are the only repairs based on attribute-value changes (into $\nn$) that are minimal under set inclusion of changes. More precisely, they are consistent, and there is not other consistent repaired version of this kind $D'$ for which $\Delta^\nn(D,D') \subsetneqq \Delta^\nn(D,D_1)$, and similarly for $D_2$. We denote this class of repairs (and the associated repair semantics) by $\nit{Srep}^{\nn}(D,\Sigma)$. \ Since $\Delta^\nn(D,D_1) \subsetneqq \Delta^\nn(D,D_3)$, $D_3 \notin \nit{Srep}^{\nn}(D,\{\kappa\})$. So, $\nit{Srep}^{\nn}(D,\{\kappa\}) = \{D_1,D_2\}$.

As with S-repairs, we can consider the subclass of repairs that minimize the number of changes, denoted $\nit{Crep}^{\nn}(D,\Sigma)$. In this example, $D_2$ is the only attribute-based cardinality repair: \  $\nit{Crep}^{\nn}($ $D,\{\kappa\}) = \{D_2\}$ \boxtheorem
\end{example}
Inspired by (\ref{eq:distG3}), we define: \vspace{-3mm}
\begin{eqnarray}
%\hspace*{-1mm}\mbox{\nit{inc-deg}}^{s,\nn,g_3\!}(D,\Sigma)  &:=&  \frac{\nit{min}\{|\Delta^\nn(D,D')|~:~ D' \in \nit{Rep}^{s,\nn}(D,D')}{|D|}\} \nonumber\\
\mbox{\nit{inc-deg}}^{c,\nn,g_3\!}(D,\Sigma)  &:=&  \frac{\nit{min}\{|\Delta^\nn(D,D')|~:~ D' \in \nit{Crep}^{\nn}(D,\Sigma)\}}{|\nit{atv}(D)|},\nonumber
\end{eqnarray}
where $\nit{atv}(D)$ is the number of values in attributes of tuples in $D$.
\begin{example}  \label{ex:cause4} (ex. \ref{ex:cause3} cont.) \ Here, $\mbox{\nit{inc-deg}}^{c,\nn,g_3\!}(D,$ $\{\kappa\}) = \frac{1}{8}$, whereas $\mbox{\nit{inc-deg}}^{c,g_3\!}(D,$\linebreak $\{\kappa\}) = \frac{1}{5}$. Under attribute-based repairs semantics, it is easy to restore consistency: only one attribute value in the database has to be changed.
\boxtheorem
\end{example}
The computation of this measure can be done on the basis of ASPs for null-based attribute repairs that were introduced in \cite{foiks18}.
\vspace{-2mm}
\section{Extensions and Discussion}
\label{sec:open}\vspace{-2mm}

We have scratched the surface of some of the problems and research directions we considered in this work. Certainly all of them deserve further investigation, most prominently, the analysis of other inconsistency measures as those in Section \ref{sec:alter} and others, and the relationships between them.  Also a deeper analysis of the incremental case (cf. Section \ref{sec:updates}) comes to mind. \ It is also left for ongoing and future research establishing a connection to the problem of computing specific repairs, and using them \cite{benny}. The same applies to the use of the inconsistency measure to explore the {\em causes for inconsistency}, in particular, to analyze how it changes when tuples or combinations thereof are removed from the database. Such an application sounds natural given the established connection between database repairs, causality and causal responsibility \cite{tocs,foiks18}.

In relation to the abstract setting of Section  \ref{sec:incd}, we could consider a class
  $\nit{Rep}^{{\sf S}^\preceq\!}(D,\Sigma)$ of {\em prioritized repairs} \cite{stawo}, and through them introduce {\em prioritized measure of inconsisrtency}.  Repair programs for the kinds of priority relations $\preceq$ investigated in \cite{stawo} could be constructed from the ASPs introduced and investigated in \cite{gebser} for capturing different optimality criteria. The repair programs could be used to specify and compute the corresponding prioritized inconsistency measure.

It is natural to think of a  principled, postulate-based approach to inconsistency measures, similar in spirit to postulates for belief-updates \cite{katsuno-mendelzon}. This has been done in logic-based knowledge representation \cite{vanina}, but as we argued before, a dedicated, specific approach for databases becomes desirable.
\ In the following we go a bit deeper into some additional open directions of research.

\vspace{-2mm}
\subsection{Incremental computation of the inconsistency degree}\label{sec:incre}\vspace{-2mm}

In relation to the analysis of changes of the inconsistency degree under updates, a deeper analysis is open, including complexity in terms of the size of the  updates. This includes fixed-parameter tractability and approximation, much in the spirit of incremental consistent query answering \cite{icdt07}.

Also algorithms for incremental computation of the inconsistency measure are need-ed. In this direction, notice that our measure can be computed through the size of a minimum vertex-cover for the set of hyperedges of the conflict-hypergraph for $D$ w.r.t. $\Sigma$. There are deterministic incremental algorithms
for computing (actually, maintaining) a $(2+\epsilon)$-approximation to a minimum vertex-cover {\em in graphs} in time $O(\nit{log}^3(n))$ for an edge- deletion or an edge-insertion, in the worst-case \cite{bhata}. Here, $n$ is the fixed number of vertices. So, only edges can be inserted or deleted. This is not exactly our situation. However, this algorithm and its properties can be adapted to our case, where edges can be added or deleted only via tuples insertions or deletions on the basis of a fixed set of DCs, which we will assume for the moment have at most two database atoms (e.g. FDs), so we have a {\em conflict-graph}.

In our setting one can consider first a fixed, finite data domain, which gives rise to a finite number of potential tuples. We can assume the set of vertices (i.e. number of tuples) has a size $n = |D| + k \times |D|$, but the latter extra vertices do not participate in any DC violation, which can  be ensured through the use of nickname predicates that are not mentioned in the DCs. Accordingly, adding a tuple outside $D$ or deleting a tuple from $D$ amounts to disabling or activating its nickname predicate, which will have the effect of creating new edges (maybe more than one) or eliminating some old edges (always at most a polynomial number of them  in $n$). After that, the above mentioned approximate algorithm for maintaining a minimum vertex-cover can be applied, as many times as edges are inserted or deleted. The size of the maintained vertex-cover  can be used to approximate the inconsistency measure with logarithmic-time for each of the updated edges.

In the case of DCs, we have hyperedges, but of bounded size, say $d$. It is likely that the approximation algorithm in \cite{bhata} can be extended to this case, but  with a $(d+\epsilon)$-approximation (as is common in the transition from graphs to hypergraphs with bounded hyperedges, e.g. see Section \ref{sec:comple}).

\vspace{-2mm}
\subsection{Sampling and sizes}\vspace{-2mm} The inconsistency measure can be seen as a form of complex aggregation in a database. As such, it becomes natural to try to approximate its value, specially in a huge database. Deterministic and randomized approximations as discussed in Section \ref{sec:comple} can be used, but adopting a statistical point of view, sampling the database to approximate the inconsistency measure looks quite appealing. The natural problem that immediately comes to mind is about the characterization and computation of the ``best" {\em statistics} defined on a sample of the database that can be used to provide a ``good" estimate of the inconsistency  measure. Also developing sampling techniques becomes crucial.

Whenever we consider sampling and estimates, {\em sizes} become relevant. In our case, relevant sizes are, apart from that of the database, the number of hyperedges in the conflict-hypergraph,  and the degrees in it of the database tuples (cf. the discussion right after Proposition \ref{prop:random}).  Both sizes are polynomial in the size of the database and the extensions of the associated sets can be defined as views over the CQs  associated to the DCs. More precisely, we can: (a) introduce tuple-identifiers (tids) for the tuples in $D$, (b) assign an order, $\prec$, to the list of predicates in the  schema; and (c)
for each DC $\kappa\!:  \neg \exists \bar{x} \Phi(\bar{x})$, with $\Phi(\bar{x})$ being the associated CQ or join, introduce a new predicate $\nit{HE}_\kappa$ for the hyperedges associated to $\kappa$. For example, if $\kappa$ is $\neg \exists \bar{x}_1\bar{x}_2\bar{x}_3 (P(\bar{x}_1) \wedge R(\bar{x}_2) \wedge S(\bar{x}_3))$, with $P \prec R \prec S$, the extension of $\nit{HE}_\kappa$ is defined (in Dalatog) by: \
$\nit{HE}_\kappa(t_1,t_2,t_3) \leftarrow P(t_1;\bar{x}_1), R(t_2;\bar{x}_2), S(t_3;\bar{x}_3)$. \ Next, on the basis of the $\nit{HE}_\kappa$ one can define a predicate collecting the neighbors of tuples, which can be used to compute or estimate the maximum degree of a tuple (the $\Delta$ mentioned after Proposition \ref{prop:random}). \
It would be interesting to investigate to what extent optimal output size bounds for the set of answers to these ``denial CQs", i.e. to the CQs $\Phi(\bar{x})$ \cite{hung},  can be taken advantage of to provide optimal estimates for the sizes of the hyperedges and tuple degrees.

\vspace{-2mm}
\subsection{Alternative inconsistency measures}\label{sec:alter}\vspace{-2mm}

 Exploring other possible inconsistency measures in our relational setting is quite an  open research direction. \ Several (in)consistency measures have been considered in knowledge representation \cite{hunter,thimm,vanina}, mostly for the propositional case or are applied with grounded first-order representations. \ It would be interesting to analyze the general properties of those measures that are closer to database applications, along the lines of
\cite{eiterMannila}; and their relationships. For each measure it becomes relevant to investigate the complexity of its computation, in particular, in data complexity (even for simple key constraints, databases may have exponentially many repairs in size of the database \cite{Bertossi2011}).

 A first observation is that, as argued in \cite{icdt07},  techniques and results for C-repairs can be extended to deal with databases whose tuples have weights, and in order to repair the aggregated
weight of removed tuples has to be a minimum.\footnote{Weighted repairs have been considered in \cite{icdt07,du,kolaitis17}.} Accordingly, $\mbox{\nit{inc-deg}}^{c,g_3\!}(D,\Sigma)$ and its results can be extended to ``weighted-repairs". \ Furthermore,
this measure, although based on tuple-deletions in the presence of DCs, can be applied with other classes of ICs, such as {\em inclusion dependencies}, and more generally,
{\em tuple-generating dependencies} (TGDs) \cite{AHV95},  if we still repair the database by tuple-deletions \cite{chomicki}. In this case, the results in Section \ref{sec:comple}  apply to TGDs since  their antecedents are treated as DCs.

We assume in the rest of this section that $\Sigma$ is a set of DCs, and the repair actions are tuple-deletions. Here below we briefly introduce a couple of alternative inconsistency measures that could be further investigated along similar lines as in the previous sections.

\vspace{2mm}%\noindent {\bf (A).} \
%\vspace{-5mm}
\begin{equation}
%\mbox{\nit{inc-deg}}^{\#}\ignore{{{\sf S},g_3\!}}(D,\Sigma) = 1- \frac{|\{D'~|~D'\subseteq D \mbox{ and } D' \models \Sigma\}|}{2^{|D|}}. \label{eq:inc}
%\mbox{\nit{inc-deg}}^{\#}(D,\Sigma) = 1- \frac{|\nit{Srep}(D)|}{2^{|D|}}. \label{eq:inc}
\hspace*{-50mm} \mbox{{\bf (A).}} \hspace{2cm} \mbox{\nit{inc-deg}}^{s,\#}(D,\Sigma) = \frac{|\nit{Srep}(D)|}{2^{|D|}}. \label{eq:inc}
\end{equation}
 Under DCs, there is always at least one S-repair (and exactly one if $D$ is already consistent or the single DC only prohibits a particular tuple); then the minimum value this measure can take is
 $\frac{1}{2^{|D|}}$. Since proper subsets of S-repairs are not S-repairs, this measure never takes the value $1$ (nor the value
$0$, as we just argued). \ Measure $\mbox{\nit{inc-deg}}^{c,\#}(D,$ $\Sigma)$, defined as in (\ref{eq:inc}) with C-repairs replacing S-repairs, does not coincide with \linebreak  $\mbox{\nit{inc-deg}}^{s,\#}(D,\Sigma)$ (in contrast with the measure in Section \ref{sec:ours}).

The denominator in (\ref{eq:inc}) may be too large. So, to obtain $0$ when the database is consistent, the measure could be modified as
\begin{equation}
\mbox{\nit{inc-deg}}^{\nit{all},\#}(D,\Sigma) := 1- \frac{|\{D'~|~D'\subseteq D \mbox{ and } D' \models \Sigma\}|}{2^{|D|}}. \label{eq:incMod}
\end{equation}
If $D$ is consistent, every subset also is, and the measure takes value $0$.

  The complexity of counting S-repairs wrt. FDs that satisfy a given Boolean conjunctive query (BCQ) was investigated in \cite{wijsenCount}. \ Depending on the syntactic form of the query,
  this can be done in polynomial time or is $\sharp P$-complete (a dichotomy);  all this in data complexity. It is easy to obtain from these results that the problem of counting the number of S-repairs wrt. key constraints can be solved in polynomial time in data complexity: simply add an atom $A$ to the database that does not participate
  in any violation and ask how many S-repairs make the (very simple) BCQ about $A$ true.

The measure in (\ref{eq:inc}) could be generalized to $\mbox{\nit{inc-deg}}^{{\sf S},\#}(D,\Sigma)$, with a generic repair semantics $\sf{S}$, by replacing $\nit{Srep}(D)$ by
$\srep$. Under some repair semantics, an inconsistent database might have no repairs, e.g. if it accepts only endogenous repairs, as in Example \ref{ex:endo}.
In this case $\mbox{\nit{inc-deg}}^{{\sf S},\#}(D,\Sigma)$ returns $0$.
So, in this case the absence of repairs is interpreted, in some sense, as perfect consistency (in contrast to the result in Example \ref{ex:endo}).

\vspace{2mm}%\noindent {\bf (B).}
%\vspace{-3mm}
\begin{equation}
\hspace*{-40mm}\mbox{{\bf (B).}} \hspace{2cm}\mbox{\nit{inc-deg}}^{s,J\!}(D,\Sigma) := 1- \frac{|\bigcap \nit{Srep}(D)|}{|D|},
\end{equation}
which is inspired by the {\em Jaccard} distance \cite{ullman}. It takes the value $0$ when $D$ is consistent, and $1$ when $\bigcap \nit{Srep}(D) = \emptyset$, i.e. when the intersection of the repairs is empty, showing that every tuple is involved in an IC violation, and nothing forces us to keep it in every repair.\footnote{An IC that forces a particular tuple to be in the database is not (logically equivalent to) a DC.}

As with (A), this measure can be generalized to $\mbox{\nit{inc-deg}}^{{\sf S},J\!}(D,\Sigma)$, with a generic repair semantics ${\sf S}$. In this case,
an inconsistent database might have no repairs (as discussed for (A) above); and, trivially, $\bigcap \srep = \bigcap \emptyset = D$; and then, \linebreak  $\mbox{\nit{inc-deg}}^{{\sf S},J\!}(D,\Sigma) = 0$. So as with (A), under this inconsistency measure the absence of repairs is interpreted as perfect consistency.

\vspace{-2mm}
\subsection{Beyond relational DBs: \ ontology-based data access}\vspace{-2mm}

Ontology-based data access (OBDA) is about accessing data from underlying sources through an ontology, most typically via queries expressed in the language of the ontology, which has access to the data through mappings \cite{xiao}. The combination of extensional database (EDB) and the ontology may become inconsistent and has to be repaired. The main approaches so far are based on (possibly virtual) changes on the EDB, mostly tuple deletions \cite{meghyn,lembo,vanina}, and consistently querying the resulting possible worlds (ontologies). Approaches to ``ontological inconsistency-tolerance" that privilege deletions of extensional tuples, and implicitly shift the culprit for inconsistency to the EDB make it reasonable to apply our inconsistency measures to the combination of extensional data and ontologies.

\vspace{-2mm}
\subsection{ASP, DBs and In-DB}\label{asp+}\vspace{-2mm}
Answer-set programming (ASP) can be seen as an extension of Datalog that supports disjunction, non-stratified negation, and constraints. Furthermore, if the semantics of ASP is applied to a Datalog program one reobtains the intended Datalog semantics. ASP has become the {\em de facto} standard language for representing and performing non-monotonic reasoning in knowledge representation.

Applying ASP to data management problems, with the database providing  the extensional data for the program,  is not only natural, but unavoidable if one wants to represent those data problems in general declarative terms, wants an exact solution, and the complexity of those problems is higher than polynomial (in data complexity) \cite{dlvDB,dlv,monica}. Actually, ASP captures problems at the second-level of the polynomial hierarchy \cite{dantsin}, and can be successfully used to specify and solve in declarative terms complex combinatorial problems. (For example, instead of following the repair-program route in Section \ref{sec:asp}, we could directly specify the hitting-sets or vertex-covers for the hyperedges in the conflict-hypergraph.)

ASP-based reasoning systems have been highly optimized \cite{brewka}, but for complexity-theoretic reasons they cannot be run inside a relational database. However, it would be really interesting to investigate, for database applications with large volumes of data, under what conditions and to what extent parts of the computation associated to the execution of an ASP can be pushed inside the database, where highly optimized join algorithms have been recently discovered and implemented \cite{hung}. In this direction there is exciting recent work on the implementation of machine learning and optimization algorithms inside the database, the {\em in-database} approach \cite{hung2}.

\vspace{-2mm}
\subsection{Tuple-level inconsistency degrees}\vspace{-2mm}

The inconsistency measure is global in that it applies to the whole database. However, one could also investigate and measure the contribution by individual tuples to the degree of inconsistency of the database. Such local measures have been investigated before in a logical setting \cite{hunterShapley}. It turns out that in our case the global inconsistency measure can be expressed in terms of the {\em responsibility} of tuples as {\em causes} for the violation of the DCs in $\Sigma$.

The connections between database causality  \cite{suciu} and database repairs were investigated in \cite{tocs}, where it is established that the {\em responsibility} of a tuple $\tau$ as a cause for $D \not \models \Sigma$ is given by:
\begin{equation}
\rho_{_{D,\Sigma}}(\tau) = \frac{1}{|D| - \mbox{max}(|S|)},
\end{equation}
where $S\subseteq D$ is an S-repair of $D$ wrt. $\Sigma$ and $\tau \notin S$ \ (but $\rho_{_{D,\Sigma}}(\tau):= 0$ if there is not such an $S$). \ Combining this with (\ref{eq:s})
 and (\ref{eq:c}), we can see that
\begin{equation}
 \mbox{\nit{inc-deg}}^{c,g_3\!}(D,\Sigma) = \frac{1}{\rho_{_{D,\Sigma}}(\tau) \times |D|}, \label{eq:final}
 \end{equation}
where $\tau$ is one and any of the {\em maximum-responsibility} tuples $\tau$ as causes for $D \not \models \Sigma$. \ We can also consider the responsibility of  tuple, $\rho_{_{D,\Sigma}}(\tau)$, as its degree of contribution to the inconsistency of the database, and those with the highest responsibility as those with a largest degree of  contribution.  According to (\ref{eq:final}), the global inconsistency measure turns out to be an aggregation over  local, tuple-level, degrees of inconsistency.

\vspace{2mm}
\noindent {\bf Acknowledgments:} \ The author has been supported by NSERC Discovery Grant \#06148. He is grateful to Jordan Li for his help with example on DLV; and to Benny Kimelfeld, Sudeepa Roy and Ester Livshits for stimulating general conversations of inconsistency measures. Excellent comments received from anonymous reviewers for a previous version of this paper are much appreciated.

\bibliographystyle{plain}

%\vspace{-4mm}

%\newpage
 \section*{Appendix A. \ An Extended  Example with DLV-Complex}\label{sec:dlv}

 In this section we retake our running example (cf. Examples \ref{ex:rep}, \ref{ex:rep2} and  \ref{ex:del}), showing how to compute repairs and inconsistency degrees by means of
 DLV-Complex \cite{dlv,calimeri08}.

The atoms in the database, with global tuple-ids, are:
{\footnotesize \begin{verbatim}
    p(1,a).    p(2,e).    q(3,a,b).   r(4,a,c).
\end{verbatim} }
The repair rules in Example \ref{ex:rep2} in their non-disjunctive versions are:
{\footnotesize \begin{verbatim}
    p_a(T,X,d)   :- p(T,X), q(T2,X,Y), not q_a(T2,X,Y,d).
    q_a(T,X,Y,d) :- q(T,X,Y), p(T2,X), not p_a(T2,X,d).

    p_a(T,X,d)   :- p(T,X), r(T2,X,Y), not r_a(T2,X,Y,d).
    r_a(T,X,Y,d) :- r(T,X,Y), p(T2,X), not p_a(T2,X,d).
\end{verbatim} }
The rules used to collect atoms in the repairs, as in Example \ref{ex:rep2}, are:
{\footnotesize \begin{verbatim}
    p_a(T,X,s)   :- p(T,X), not p_a(T,X,d).
    q_a(T,X,Y,s) :- q(T,X,Y), not q_a(T,X,Y,d).
    r_a(T,X,Y,s) :- r(T,X,Y), not r_a(T,X,Y,d).
\end{verbatim} }
The following rules retrieve the tids of deleted tuples:
{\footnotesize \begin{verbatim}
    del(T) :- p_a(T,X,d).
    del(T) :- q_a(T,X,Y,d).
    del(T) :- r_a(T,X,Y,d).
\end{verbatim} }
\noindent The following rules compute, in this order and per repair:  the number of deleted tuples (per repair), the cardinalities of  the original tables, the number of tuples in the database, the cardinality of each repaired table, the cardinality of the repair, and, finally,  the number of tuples in the difference between the original instance and the repair.
{\footnotesize \begin{verbatim}
    #maxint = 100.
    numDel(N) :- #int(N), #count{T: del(T)} = N.
    cardPred(p,N) :- #int(N), #count{T : p(T,X)} = N.
    cardPred(q,N) :- #int(N), #count{T : q(T,X,Y)} = N.
    cardPred(r,N) :- #int(N), #count{T : r(T,X,Y)} = N.
    cardDB(N) :- #sum{X,P : cardPred(P,X)} = N.
    cardRep(p,N) :- #int(N), #count{T : p_a(T,X,s)} = N.
    cardRep(q,N) :- #int(N), #count{T : q_a(T,X,Y,s)} = N.
    cardRep(r,N) :- #int(N), #count{T : r_a(T,X,Y,s)} = N.
    cardRepDB(N) :- #int(N), #sum{X,P : cardRep(P,X)} = N.
    dist(N) :- #int(N), cardDB(A), cardRepDB(B), N = A - B.
\end{verbatim} }
\noindent Running the program we obtain two stable models, corresponding to the two S-repairs in Example \ref{ex:rep}; each of them showing the (unnormalized) distance to the original instance, namely 2 and 1, resp.:

\newpage {\footnotesize \begin{verbatim}
    DLV [build BEN+ODBC/Dec 17 2012   gcc 4.6.1]

    {p(1,a), p(2,e), q(3,a,b), r(4,a,c), cardPred(p,2),
    cardPred(q,1), cardPred(r,1), cardDB(4), q_a(3,a,b,d),
    r_a(4,a,c,d), p_a(1,a,s), p_a(2,e,s), del(3), del(4),
    cardRep(p,2), cardRep(q,0), cardRep(r,0), cardRepDB(2),
    numDel(2), dist(2)}

    {p(1,a), p(2,e), q(3,a,b), r(4,a,c), cardPred(p,2),
    cardPred(q,1), cardPred(r,1), cardDB(4), p_a(1,a,d),
    q_a(3,a,b,s), r_a(4,a,c,s), p_a(2,e,s), del(1),
    cardRep(p,1), cardRep(q,1), cardRep(r,1), cardRepDB(3),
    numDel(1),    dist(1)}
\end{verbatim} }
\noindent The second model (repair) is the only C-repair, which is the one giving the minimum distance, $1$.
\ If we are interested only in the possible distances with origin in the different repairs, we can add a query about them (It can be included at the end of the program file). The answers under the {\em possible} or {\em brave semantics}
will be those obtained from some repair:\footnote{Having the query in the program file (say `progFile"), after the program, this is done by running from the
DLV command line: ``dlv -brave progFile". \ For the cautions (or certain) answers, i.e. those true in {\em all} repairs, we would use ``dlv -cautions progFile".}
{\footnotesize \begin{verbatim}
    dist(X)?
    1
    2
\end{verbatim} }
\noindent From this we obtain $1$ as the minimum distance. This off-line comparison of distances, either through the query results or inspection of the models (as above),  can be avoided by adding
to the program above a weak constraint (WC) aiming at minimizing the number of deleted tuples:
{\footnotesize \begin{verbatim}
    :~ del(T).
\end{verbatim} }
\noindent The output shows only the C-repair including the unnormalized distance to the original instance, namely $1$, and the cost as the number of violations of the only WC:
{\footnotesize \begin{verbatim}
    Best model: {p(1,a), p(2,e), q(3,a,b), r(4,a,c), cardPred(p,2),
    cardPred(q,1), cardPred(r,1), cardDB(4), p_a(1,a,d),
    q_a(3,a,b,s), r_a(4,a,c,s), p_a(2,e,s), del(1), cardRep(p,1),
    cardRep(q,1), cardRep(r,1), cardRepDB(3), numDel(1), dist(1)}

    Cost ([Weight:Level]): <[1:1]>
\end{verbatim} }

\newpage
\section*{Appendix B. \ Proofs of Results}\label{app:proofs}

\defproof{Theorem \ref{thm:fp}}{Computing $\mbox{\nit{inc-deg}}^{c,g_3\!}(D,\Sigma)$ is basically about computing \linebreak $\nit{max}\{ |D'| : D' \in \crep  \}$. Since all C-repairs have the same size, we need to compute the size of a C-repair wrt. DCs.
 This problem is $\nit{FP}^{\nit{NP(log(n))}}$-complete in data complexity\cite[theo. 3]{icdt07}.}\\

  \defproof{Lemma \ref{lemma:max}}{Consider the relational predicate $C(\nit{clause},\nit{variable},\nit{sign})$, with the FDs: $\nit{clause}$ $\rightarrow$ $\nit{variable}$, and $\nit{variable}$ $\rightarrow$ $\nit{sign}$.

  Consider now an instance for the 3-SAT problem, as a propositional formula $\psi$ in CNF over the propositional variables $p_1, p_2, \ldots$.  Assume that $\psi$ is of the form $c_1 \wedge \cdots \wedge c_m$, with
  each $c_i$ a disjunction of three literals, i.e. propositional variables or negations thereof. We may assume that each $c_i$ does not contain a variable and its negation.

 From $\psi$ we construct an instance $D$ for this schema, as follows. For each clause $c_i$ and propositional variable $p_j$ in it, create the tuple $C(c_i,p_j, \pm)$, with $-$ if $p_j$ appears negated and $+$, otherwise.

 Instance $D$ is inconsistent wrt. $\Sigma$ (except in the extreme and trivial case where each clause contains a single and distinct literal), and $\nit{CG}(D,\Sigma)$, that has the tuples as vertices, contains an edge between $C(c_i,p_j,s)$ and  $C(c_k,p_l,s')$ iff (a) $p_j = p_l$ and $s\neq s'$, or (b) $c_i = c_k$ and $p_j \neq p_l$.

  Consider now the complement of the conflict graph, $\nit{CG}^c(D,\Sigma)$. The tuples are the same, but there is an edge between $C(c_i,p_j,s)$ and  $C(c_k,p_l,s')$ iff $c_i \neq c_k$ and $p_j\neq p_l$, or $c_i \neq c_k$ and
  $s = s'$. Since in this graph there are never two nodes of the form $C(c,p,-1)$ and $C(c,p,+1)$, it is isomorphic to the graph $\mc{G}$ with nodes $C(c,p)$, for some $C(c,p,s) \in \nit{CG}^c(D,\Sigma)$, and  with the edges inherited from  $\nit{CG}^c(D,\Sigma)$. This graph $\mc{G}$ is the one that one builds to reduce $\psi$ to a graph \cite[theo. 10.5]{ullmanAlg}, in such a way that $\psi$ has $k$ clauses satisfied iff $\mc{G}$ has a clique of size $k$.\footnote{For a reduction from SAT to the Independent Set problem, see \cite[theo. 9.4]{papa}.} Now,
  $\mc{G}$ has a clique of size $k$ iff its complement $\nit{CG}(D,\Sigma)$ has an independent set of size $k$.\ignore{; that is, to establish that the $k$-clique problem is NP-hard.} Since $k$-satisfiability of 3-CNF formulas is \nit{NP}-complete \cite[theo. 9.2]{papa}, we obtain the result.}

  \begin{example} Consider the formula $\psi\!: \ c_1 \wedge c_2 \wedge c_3$, with $c_1\!: \ (p_1 \vee \neg p_2), \ c_2\!: \ (p_2 \vee \neg p_3), \ c_3\!: \ (p_3 \vee \neg p_1)$. The conflict graph $\nit{CG}(D,\Sigma)$ is shown on the left-hand side  below, and its complement graph, $\nit{CG}^c(D,\Sigma)$, on the right-hand side.

\begin{multicols}{2}
  \begin{center}
  \includegraphics[width=6cm]{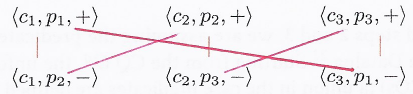}
  \end{center}

 \ignore{
  \begin{eqnarray*}
  \langle c_1, p_1, +\rangle\hspace{1cm}&\langle c_2, p_2, +\rangle&\hspace{1cm}\langle c_3, p_3, +\rangle\\
  \red{|} \hspace{2cm}& \red{|} &\hspace{2cm} \red{|}\\
  \langle c_1, p_2, -\rangle\hspace{1cm}&\langle c_2, p_3, -\rangle&\hspace{1cm}\langle c_3, p_1, -\rangle
  \end{eqnarray*} }

%  The complement of this graph, $\nit{CG}^c(D,\Sigma)$, is:

  \begin{center}
  \includegraphics[width=6cm]{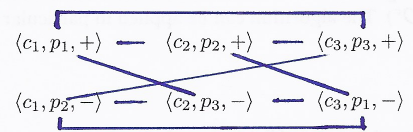}
  \end{center}
\end{multicols}

 \ignore{ \begin{eqnarray*}
  \langle c_1, p_1, +\rangle\hspace{0.2cm}\blue{--}\hspace{0.2cm}&\langle c_2, p_2, +\rangle&\hspace{0.2cm}\blue{--}\hspace{0.2cm}\langle c_3, p_3, +\rangle\\
  \\
  \langle c_1, p_2, -\rangle\hspace{0.2cm}\blue{--}\hspace{0.2cm}&\langle c_2, p_3, -\rangle&\hspace{0.2cm}\blue{--}\hspace{0.2cm}\langle c_3, p_1, -\rangle
  \end{eqnarray*}  }

 The maximum size of an independent set in $\nit{CG}(D,\Sigma)$  is the same as the size of maximal clique in $\nit{CG}^c(D,\Sigma)$, which is $3$, and is also the maximum number of simultaneously satisfiable clauses $c_i$ in $\psi$ (and then the formula is satisfiable). \boxtheorem
  \end{example}

\defproof{Corollary \ref{cor:fds}}{For the schema  and instance $D$ as in the lemma, there is a C-repair of size at least $k$ iff in the conflict graph there is an independent set of size at least $k$.}\\

\defproof{Theorem \ref{thm:fixed}}{Membership follows from Corollary \ref{cor:fds} in combination with binary search for computing the size of C-repair, which can be used to compute the measure.
\ Completeness follows from the reduction from maximum-number of clause-satisfaction  for SAT  to maximum-size of a clique in the complement of  $\nit{CG}(D,\Sigma)$. The former problem  is $\nit{FP}^{\nit{NP(log(n))}}$-complete  \cite[theo. 2.2]{krentel}.}\\

 \defproof{Theorem \ref{thm:det}}{We appeal again to the conflict-hypergraph in Remark \ref{rem:cg}. The result is obtained from a polynomial-time  approximation -via integer programming relaxation into  linear programming- to  the
(size of a) minimum-vertex cover problem in a hypergraph whose hyperedges are bounded above in size by a number $d$. There is a $d$-ratio approximation algorithm (\cite[chap. 3]{hochbaum} and \cite{yehuda}).}\\

\defproof{Proposition \ref{prop:cont1}}{Let us assume $k$ out of the $m = \epsilon \times |D|$ new tuples participate in new violations, in combination with new or old tuples, i.e. they appear in subset-minimal hyperedges for $D'$. If we delete these $k$ tuples, every C-repair for $D$ is also a C-repair for $D$ plus the $m-k$ non-violating new tuples. Accordingly, C-repairs for $D'$ are obtained by deleting at most $k$ tuples plus  those deleted to obtain a C-repair for $D$. Then,

\vspace{2mm}$\frac{\nit{min} \{|D' \smallsetminus D''|~:~ D'' \in \nit{Crep}(D',\Sigma)\}}{|D| + m} \leq \frac{\nit{min} \{|D \smallsetminus D''|~:~ D'' \in \nit{Crep}(D,\Sigma)  \} + k}{|D| + m} \leq$\\ \hspace*{8.3cm}$\mbox{\nit{inc-deg}}^{c,g_3}(D,\Sigma) + \frac{m}{|D| + m}$.

For the second part, let $D^{\star\star}$ be a C-repair for $D$, and $D^\star$ a C-repair for $D' = D \cup D_k$. \ Now, $D^\star \smallsetminus D_k$ is a consistent sub-instance of $D$. Since, $D^{\star\star}$ is a C-repair of $D$:
\begin{eqnarray*}
\mbox{\nit{inc-deg}}^{c,g_3}(D,\Sigma) &=&
\frac{|D\smallsetminus D^{\star \star}|}{|D|} \leq \frac{|D\smallsetminus (D^\star \smallsetminus D_k)|}{|D|} = \frac{|D\smallsetminus D^\star|}{|D|} =\\ && \frac{|(D' \smallsetminus D_k) \smallsetminus D^\star|}{|D|}
= \frac{|(D' \smallsetminus (D_k \cup  D^\star)|}{|D|} \leq \frac{|D' \smallsetminus D^\star|}{|D|} =\\&& \frac{|D' \smallsetminus D^\star|}{(1-\epsilon) \times |D'|} = \frac{1}{1-\epsilon} \times \mbox{\nit{inc-deg}}^{c,g_3}(D',\Sigma).
\end{eqnarray*}
}\\

\defproof{Proposition \ref{prop:cont2}}{ Let $D' = D \smallsetminus D_k$, with $|D_k| = k = \epsilon \times |D|$. So, $|D'| = (1-\epsilon) \times |D|$. \ Let $D_1$ be a C-repair for $D'$, then $\mbox{\nit{inc-deg}}^{c,g_3}(D',\Sigma) = \frac{|D' \smallsetminus D_1|}{|D'|}$. Since $D_1$ is consistent and contained in $D$, it is also a repair for $D$, but possibly non-maximum in size. Then, with $D^*$ a C-repair for $D$, $\mbox{\nit{inc-deg}}^{c,g_3}(D',\Sigma) \leq \frac{|D' \smallsetminus D^*|}{|D'|} = \frac{|D \smallsetminus D^*|}{(1-\epsilon) \times |D|} = \frac{1}{1-\epsilon} \times \mbox{\nit{inc-deg}}^{c,g_3}(D,\Sigma)$.

For the second part, let $D^{\star\star}$ be a C-repair for $D$, $D^\star$ a C-repair for $D'$, and $D_k = D_{k'} \cup D_{k-k'}$, $0 \leq k' \leq k$, be a partition of $D_k$ into the tuples that participate in DC violations in $D$, and those that do not. Then, $D^\star \cup D_{k-k'}$ is an S-repair for $D$. Then,
\begin{eqnarray*}
\mbox{\nit{inc-deg}}^{c,g_3}(D,\Sigma) &=& \frac{|D \smallsetminus D^{\star\star}|}{|D|} \leq \frac{|D \smallsetminus (D^\star \cup D_{k-k'})|}{|D|} = \frac{(D' \smallsetminus D^\star) \cup D_{k'}|}{|D|} \\
&=& \frac{|D' \smallsetminus D^\star| + |D_{k'}|}{|D|} \leq \frac{|D' \smallsetminus D^\star|}{(1-\epsilon) \times |D'|} + \frac{|D_{k'}|}{|D|}\\
 &\leq& \frac{1}{(1-\epsilon)} \times \mbox{\nit{inc-deg}}^{c,g_3}(D',\Sigma) + \epsilon.
\end{eqnarray*}
When $D_{k'} = \emptyset$, the last term disappears.}\\

\defproof{Theorem \ref{theo:fpt}}{The conflict-hypergraph $\nit{CG}(D,\Sigma)$ in Remark \ref{rem:cg} has its hyperedges bounded above in size by $d$. The C-repairs are in one-to-one correspondence with the minimum-vertex covers: the deletion of such a vertex cover produces a C-repair, because this eliminates one tuple from each conflict and so restores consistency in a minimum way. We are interested in determining the size of a minimum vertex cover. \ Then, this is a case of the so-called \emph{d-hitting set problem}, consisting
in finding the size of a minimum hitting set for an hypergraph
with hyperedges bounded in size by $d$.

It is known that the problem of determining if a graph of size $n$ has a vertex cover of size not larger than $k$ is
$\nit{FPT}$ with parameter $k$ \cite{chen,niedermeier}, that is, there is a decision algorithm that runs $O(C^k + kn)$. This is exponential in parameter $k$, but linear in $n$. In our case, we have an initial graph of size $N$, without edges, plus $m$ additional nodes that can have edges between them or with pre-existing nodes. By binary search on $m$, we can determine the size of a minimum vertex cover for the graph with $N + m$ nodes in time
bounded above by $O(\nit{log}(m) \times (C^m + m N))$. This value can be used to easily compute the inconsistency measure. This argument also applies to hypergraphs with $d$-bounded edges, in which case the constant $C$ depends on $d$ \cite{niedermeier}. }

\ignore{
\section*{Appendix C. \ Extensions and Discussion}
\label{sec:open}

We have scratched the surface of some of the problems and research directions we considered in this work. Certainly all of them deserve further investigation, most prominently, the analysis of other inconsistency measures as those in Section \ref{sec:alter} and others, and the relationships between them.  Also a deeper analysis of the incremental case (cf. Section \ref{sec:updates}) comes to mind. \ It is also left for ongoing and future research establishing a connection to the problem of computing specific repairs, and using them \cite{benny}. The same applies to the use of the inconsistency measure to explore the {\em causes for inconsistency}, in particular, to analyze how it changes when tuples or combinations thereof are removed from the database. Such an application sounds natural given the established connection between database repairs, causality and causal responsibility \cite{tocs,foiks18}.

It is natural to think of a  principled, postulate-based approach to inconsistency measures, similar in spirit to postulates for belief-updates \cite{katsuno-mendelzon}. This has been done in logic-based knowledge representation \cite{vanina}, but as we argued before, a dedicated, specific approach for databases becomes desirable.
\ In the following we go a bit deeper into some additional open directions of research.

\subsection*{Incremental computation of the inconsistency degree}\label{sec:incre}

In relation to the analysis of changes of the inconsistency degree under updates, a deeper analysis is open, including complexity in terms of the size of the  updates. This includes fixed-parameter tractability and approximation, much in the spirit of incremental consistent query answering \cite{icdt07}.

Also algorithms for incremental computation of the inconsistency measure are needed. In this direction, notice that our measure can be computed through the size of a minimum vertex-cover for the set of hyperedges of the conflict-hypergraph for $D$ w.r.t. $\Sigma$. There are deterministic incremental algorithms
for computing (actually, maintaining) a $(2+\epsilon)$-approximation to a minimum vertex-cover {\em in graphs} in time $O(\nit{log}^3(n))$ for an edge- deletion or an edge-insertion, in the worst-case \cite{bhata}. Here, $n$ is the fixed number of vertices. So, only edges can be inserted or deleted. This is not exactly our situation. However, this algorithm and its properties can be adapted to our case, where edges can be added or deleted only via tuples insertions or deletions on the basis of a fixed set of DCs, which we will assume for the moment have at most two database atoms (e.g. FDs), so we have a {\em conflict-graph}.

In our setting one can consider first a fixed, finite data domain, which gives rise to a finite number of potential tuples. We can assume the set of vertices (i.e. number of tuples) has a size $n = |D| + k \times |D|$, but the latter extra vertices do not participate in any DC violation, which can  be ensured through the use of nickname predicates that are not mentioned in the DCs. Accordingly, adding a tuple outside $D$ or deleting a tuple from $D$ amounts to disabling or activating its nickname predicate, which will have the effect of creating new edges (maybe more than one) or eliminating some old edges (always at most a polynomial number of them  in $n$). After that, the above mentioned approximate algorithm for maintaining a minimum vertex-cover can be applied, as many times as edges are inserted or deleted. The size of the maintained vertex-cover  can be used to approximate the inconsistency measure with logarithmic-time for each of the updated edges.

In the case of DCs, we have hyperedges, but of bounded size, say $d$. It is likely that the approximation algorithm in \cite{bhata} can be extended to this case, but  with a $(d+\epsilon)$-approximation (as is common in the transition from graphs to hypergraphs with bounded hyperedges, e.g. see Section \ref{sec:comple}).

\subsection*{Sampling and sizes} The inconsistency measure can be seen as a form of complex aggregation in a database. As such, it becomes natural to try to approximate its value, specially in a huge database. Deterministic and randomized approximations as discussed in Section \ref{sec:comple} can be used, but adopting a statistical point of view, sampling the database to approximate the inconsistency measure looks quite appealing. The natural problem that immediately comes to mind is about the characterization and computation of the ``best" {\em statistics} defined on a sample of the database that can be used to provide a ``good" estimate of the inconsistency  measure. Also developing sampling techniques becomes crucial.

Whenever we consider sampling and estimates, {\em sizes} become relevant. In our case, relevant sizes are, apart from that of the database, the number of hyperedges in the conflict-hypergraph,  and the degrees in it of the database tuples (cf. the discussion right after Proposition \ref{prop:random}).  Both sizes are polynomial in the size of the database and the extensions of the associated sets can be defined as views over the CQs  associated to the DCs. More precisely, we can: (a) introduce tuple-identifiers (tids) for the tuples in $D$, (b) assign an order, $\prec$, to the list of predicates in the  schema; and (c)
for each DC $\kappa\!:  \neg \exists \bar{x} \Phi(\bar{x})$, with $\Phi(\bar{x})$ being the associated CQ or join, introduce a new predicate $\nit{HE}_\kappa$ for the hyperedges associated to $\kappa$. For example, if $\kappa$ is $\neg \exists \bar{x}_1\bar{x}_2\bar{x}_3 (P(\bar{x}_1) \wedge R(\bar{x}_2) \wedge S(\bar{x}_3))$, with $P \prec R \prec S$, the extension of $\nit{HE}_\kappa$ is defined (in Dalatog) by: \
$\nit{HE}_\kappa(t_1,t_2,t_3) \leftarrow P(t_1;\bar{x}_1), R(t_2;\bar{x}_2), S(t_3;\bar{x}_3)$. \ Next, on the basis of the $\nit{HE}_\kappa$ one can define a predicate collecting the neighbors of tuples, which can be used to compute or estimate the maximum degree of a tuple (the $\Delta$ mentioned after Proposition \ref{prop:random}). \
It would be interesting to investigate to what extent optimal output size bounds for the set of answers to these ``denial CQs", i.e. to the CQs $\Phi(\bar{x})$ \cite{hung},  can be taken advantage of to provide optimal estimates for the sizes of the hyperedges and tuple degrees.

\subsection*{Alternative inconsistency measures}\label{sec:alter}

 Exploring other possible inconsistency measures in our relational setting is quite an  open research direction. \ Several (in)consistency measures have been considered in knowledge representation \cite{hunter,thimm,vanina}, mostly for the propositional case or are applied with grounded first-order representations. \ It would be interesting to analyze the general properties of those measures that are closer to database applications, along the lines of
\cite{eiterMannila}; and their relationships. For each measure it becomes relevant to investigate the complexity of its computation, in particular, in data complexity (even for simple key constraints, databases may have exponentially many repairs in the size of the database \cite{Bertossi2011}).

 A first observation is that, as argued in \cite{icdt07},  techniques and results for C-repairs can be extended to deal with databases whose tuples have weights, and in order to repair the aggregated
weight of removed tuples has to be a minimum.\footnote{Weighted repairs have been considered in \cite{icdt07,du,kolaitis17}.} Accordingly, $\mbox{\nit{inc-deg}}^{c,g_3\!}(D,\Sigma)$ and its results can be extended to ``weighted-repairs". \ Furthermore,
this measure, although based on tuple-deletions in the presence of DCs, can be applied with other classes of ICs, such as {\em inclusion dependencies}, and more generally,
{\em tuple-generating dependencies} (TGDs) \cite{AHV95},  if we still repair the database by tuple-deletions \cite{chomicki}. In this case, the results in Section \ref{sec:comple}  apply to TGDs since  their antecedents are treated as DCs.

We assume in the rest of this section that $\Sigma$ is a set of DCs, and the repair actions are tuple-deletions. Here below we briefly introduce a couple of alternative inconsistency measures that could be further investigated along similar lines as in the previous sections.

%\vspace{2mm}\noindent {\bf (A).} \
\vspace{-5mm}\begin{equation}
%\mbox{\nit{inc-deg}}^{\#}\ignore{{{\sf S},g_3\!}}(D,\Sigma) = 1- \frac{|\{D'~|~D'\subseteq D \mbox{ and } D' \models \Sigma\}|}{2^{|D|}}. \label{eq:inc}
%\mbox{\nit{inc-deg}}^{\#}(D,\Sigma) = 1- \frac{|\nit{Srep}(D)|}{2^{|D|}}. \label{eq:inc}
\hspace*{-5mm} \mbox{{\bf (A).}} \hspace{1cm} \mbox{\nit{inc-deg}}^{s,\#}(D,\Sigma) = \frac{|\nit{Srep}(D)|}{2^{|D|}}. \label{eq:inc}
\end{equation}
 Under DCs, there is always at least one S-repair (and exactly one if $D$ is already consistent or the single DC only prohibits a particular tuple); then the minimum value this measure can take is
 $\frac{1}{2^{|D|}}$. Since proper subsets of S-repairs are not S-repairs, this measure never takes the value $1$ (nor the value
$0$, as we just argued). \ Measure $\mbox{\nit{inc-deg}}^{c,\#}(D,$ $\Sigma)$, defined as in (\ref{eq:inc}) with C-repairs replacing S-repairs, does not coincide with \linebreak  $\mbox{\nit{inc-deg}}^{s,\#}(D,\Sigma)$ (in contrast with the measure in Section \ref{sec:ours}).

The denominator in (\ref{eq:inc}) may be too large. So, to obtain $0$ when the database is consistent, the measure could be modified as
\begin{equation}
\mbox{\nit{inc-deg}}^{\nit{all},\#}(D,\Sigma) := 1- \frac{|\{D'~|~D'\subseteq D \mbox{ and } D' \models \Sigma\}|}{2^{|D|}}. \label{eq:incMod}
\end{equation}
If $D$ is consistent, every subset also is, and the measure takes value $0$.

  The complexity of counting S-repairs wrt. FDs that satisfy a given Boolean conjunctive query (BCQ) was investigated in \cite{wijsenCount}. \ Depending on the syntactic form of the query,
  this can be done in polynomial time or is $\sharp P$-complete (a dichotomy);  all this in data complexity. It easy to obtain from these results that the problem of counting the number of S-repairs wrt. key constraints can be solved in polynomial time in data complexity: simply add an atom $A$ to the database that does not participate
  in any violation and ask how many S-repairs make the (very simple) BCQ about $A$ true.

The measure in (\ref{eq:inc}) could be generalized to $\mbox{\nit{inc-deg}}^{{\sf S},\#}(D,\Sigma)$, with a generic repair semantics $\sf{S}$, by replacing $\nit{Srep}(D)$ by
$\srep$. Under some repair semantics, an inconsistent database might have no repairs, e.g. if it accepts only endogenous repairs, as in Example \ref{ex:endo}.
In this case $\mbox{\nit{inc-deg}}^{{\sf S},\#}(D,\Sigma)$ returns $0$.
So, in this case the absence of repairs is interpreted, in some sense, as perfect consistency (in contrast to the result in Example \ref{ex:endo}).

%\vspace{2mm}\noindent {\bf (B).}
\vspace{-3mm}\begin{equation}
\hspace*{-5mm}\mbox{{\bf (B).}} \hspace{1cm}\mbox{\nit{inc-deg}}^{s,J\!}(D,\Sigma) := 1- \frac{|\bigcap \nit{Srep}(D)|}{|D|},
\end{equation}
which is inspired by the {\em Jaccard} distance \cite{ullman}. It takes the value $0$ when $D$ is consistent, and $1$ when $\bigcap \nit{Srep}(D) = \emptyset$, i.e. when the intersection of the repairs is empty, showing that every tuple is involved in an IC violation, and nothing forces us to keep it in every repair.\footnote{An IC that forces a particular tuple to be in the database is not (logically equivalent to) a DC.}

As with (A), this measure can be generalized to $\mbox{\nit{inc-deg}}^{{\sf S},J\!}(D,\Sigma)$, with a generic repair semantics ${\sf S}$. In this case,
an inconsistent database might have no repairs (as discussed for (A) above); and, trivially, $\bigcap \srep = \bigcap \emptyset = D$; and then, \linebreak  $\mbox{\nit{inc-deg}}^{{\sf S},J\!}(D,\Sigma) = 0$. So as with (A), under this inconsistency measure the absence of repairs is interpreted as perfect consistency.

\subsection*{Beyond relational DBs: \ ontology-based data access}

Ontology-based data access (OBDA) is about accessing data from underlying sources through an ontology, most typically via queries expressed in the language of the ontology, which has access to the data through mappings \cite{xiao}. The combination of extensional database (EDB) and the ontology may become inconsistent and has to be repaired. The main approaches so far are based on (possibly virtual) changes on the EDB, mostly tuple deletions \cite{meghyn,lembo,vanina}, and consistently querying the resulting possible worlds (ontologies). Approaches to ``ontological inconsistency-tolerance" that privilege deletions of extensional tuples, and implicitly shift the culprit for inconsistency to the EDB make it reasonable to apply our inconsistency measures to the combination of extensional data and ontologies.

\subsection*{ASP, DBs and In-DB}\label{asp+}
Answer-set programming (ASP) can be seen as an extension of Datalog that supports disjunction, non-stratified negation, and constraints. Furthermore, if the semantics of ASP is applied to a Datalog program one reobtains the intended Datalog semantics. ASP has become the {\em de facto} standard language for representing and performing non-monotonic reasoning in knowledge representation.

Applying ASP to data management problems, with the database providing  the extensional data for the program,  is not only natural, but unavoidable if one wants to represent those data problems in general declarative terms, wants an exact solution, and the complexity of those problems is higher than polynomial (in data complexity) \cite{dlvDB,dlv,monica}. Actually, ASP captures problems at the second-level of the polynomial hierarchy \cite{dantsin}, and can be successfully used to specify and solve in declarative terms complex combinatorial problems. (For example, instead of following the repair-program route in Section \ref{sec:asp}, we could directly specify the hitting-sets or vertex-covers for the hyperedges in the conflict-hypergraph.)

ASP-based reasoning systems have been highly optimized \cite{brewka}, but for complexity-theoretic reasons they cannot be run inside a relational database. However, it would be really interesting to investigate, for database applications with large volumes of data, under what conditions and to what extent parts of the computation associated to the execution of an ASP can be pushed inside the database, where highly optimized join algorithms have been recently discovered and implemented \cite{hung}. In this direction there is exciting recent work on the implementation of machine learning and optimization algorithms inside the database, the {\em in-database} approach \cite{hung2}.
}

\end{document}